\documentclass{lmcs}
\pdfoutput=1

\usepackage{silence}
\usepackage{lastpage}
\lmcsdoi{15}{4}{13}
\lmcsheading{}{\pageref{LastPage}}{}{}%
{Mar.~29,~2018}{Dec.~09,~2019}{}

\usepackage{etex}

\usepackage[english]{babel}
\usepackage[utf8]{inputenc}

\usepackage{array}

\usepackage[all]{xy}

\usepackage{tikz}
\usetikzlibrary{arrows,automata}

\usepackage{cmll}

\usepackage{latexsym}
\usepackage{amsmath}
\usepackage{amssymb}
\usepackage{stmaryrd}
\usepackage{mathrsfs}
\usepackage{bm}
\usepackage{bussproofs}
\EnableBpAbbreviations%

\usepackage{url}

\usepackage{ifpdf}

\usepackage{macros}

\renewcommand\vec[1]{\overline{#1}}

\begin{document}

\author{Pierre Pradic\rsuper{{a,b}}}
\address{\lsuper{a}Univ Lyon, EnsL, UCBL, CNRS,  LIP, F-69342, LYON Cedex 07, France}
\email{\{pierre.pradic,colin.riba\}@ens-lyon.fr}
\address{\lsuper{b}University of Warsaw, Faculty of Mathematics, Informatics and Mechanics}

\author{Colin Riba\rsuper{a}}

\title[A Curry-Howard Approach to Church's Synthesis]
{A Curry-Howard Approach to Church's Synthesis\rsuper*}
\titlecomment{\lsuper*This work was partially supported by
the ANR-14-CE25-0007 - RAPIDO and the ANR-BLANC-SIMI-2-2011 - RECRÉ. 
This paper is an extended version of the conference article~\cite{pr17fscd}.}

\maketitle

\begin{abstract}
Church's synthesis problem asks whether there exists a
finite-state stream transducer satisfying a given input-output specification.
For specifications written in Monadic Second-Order Logic ($\MSO$) over infinite words,
Church's synthesis can theoretically be solved algorithmically
using automata and games.
We revisit Church's synthesis via the Curry-Howard correspondence by
introducing $\SMSO$,
an intuitionistic variant of $\MSO$ over infinite words,
which is shown to be sound and complete \wrt\@ synthesis
thanks to an automata-based realizability model.
\end{abstract}

\section{Introduction}%
\label{sec:intro}

\noindent
A stream function $F : \Sigma^\omega \to \Gamma^\omega$
is \emph{synchronous} (or \emph{causal}) if it can produce
a prefix of length $n$ of its output from a prefix of length $n$
of its input:
\[
\tag{for $\seq,\seqbis \in \Sigma^\omega$} 
\seq(0). \cdots .\seq(n-1) = \seqbis(0). \cdots .\seqbis(n-1) 
\quad\imp\quad
F(\seq)(n) = F(\seqbis)(n)
\]

\noindent
A synchronous function is \emph{finite-state} if it is induced by
a deterministic letter-to-letter
stream transducer
(or \emph{deterministic Mealy machine}, DMM).
Church's synthesis~\cite{church57ssisl} consists in the automatic extraction
of DMMs from input-output specifications,
typically presented as
closed formulae
of the form
\begin{equation}
\label{eq:intro:spec}
\forall {X \in \Sigma^\omega}\, \exists {Y \in \Gamma^\omega}\, \varphi(X;Y)
\end{equation}
where $\varphi$ is a formula of
some subsystem of \emph{Monadic Second-Order Logic} ($\MSO$) over $\omega$-words.
A specification as in (\ref{eq:intro:spec})
is realized in the sense of Church by a (finite-state)
synchronous $F : \Sigma^\omega \to \Gamma^\omega$ when
$\varphi(\seq,F(\seq))$ holds for all $\seq \in \Sigma^\omega$.

$\MSO$ over $\omega$-words is a decidable logic
(B{\"u}chi's Theorem~\cite{buchi62lmps})
which subsumes
logics used in verification such as $\LTL$
(see \eg~\cite{thomas97handbook,pp04book,vw08chapter}).
Traditional approaches to synthesis (see \eg~\cite{thomas08npgi,thomas09fossacs})
are based, via McNaughton's Theorem~\cite{mcnaughton66ic},
on the translation of $\MSO$-formulae to \emph{deterministic}
automata on $\omega$-words
(such as \emph{Muller} or \emph{parity} automata).\footnote{A solution is also
possible via tree automata~\cite{rabin72} (see also~\cite{kpv06cav,thomas09fossacs}).}
These automata are then turned into infinite two-player sequential games
on finite graphs, in which the \emph{Opponent}
($\forall$bélard) plays input letters to which the \emph{Proponent}
($\exists$loïse) replies with output letters.
Solutions to Church's synthesis are then given by
the Büchi-Landweber Theorem~\cite{bl69tams},
which states that in such games, (exactly) one of the two players
has a finite-state winning strategy
(\ie\@ a strategy which only uses a finite memory).

Fully automatic approaches to synthesis suffer from
prohibitively high computational costs, essentially for the following two reasons.
First, the translation of $\MSO$-formulae
to automata is non-elementary (see \eg~\cite{gtw02alig}),
and McNaughton's Theorem involves
a non-trivial powerset construction
(such as \emph{Safra construction},
see~\eg~\cite{thomas97handbook,gtw02alig,pp04book,vw08chapter}).
Second, similarly as with other automatic verification techniques based on
model checking, the solution of parity games ultimately relies on exhaustive
state exploration.
While they have had (and still have) considerable success for verifying concurrency
properties, such techniques
hardly managed up to now
to give
practical algorithms for the synthesis of large scale systems
(even for fragments of $\LTL$, see~\eg~\cite{bjpps12jcss}).

In this work, we propose a Curry-Howard approach to Church's synthesis.
The Curry-Howard correspondence asserts that, given a suitable
proof system, any proof therein can be interpreted as a program.
This interpretation of proofs as programs
(as well as the soundness of many type systems)
can be formalized using the technique of \emph{realizability},
which tells
how to read a formula from the logic as a specification for a program.
More precisely, realizability can be seen as a relation
between programs (the realizers) and formulae,
usually defined by induction on the latter (see \eg~\cite{su06book,kohlenbach08book}).
Typical clauses state \eg\@ that realizers of conjunctions
$\varphi_1 \land \varphi_2$ are pairs $\pair{\run_1,\run_2}$
consisting of a realizer $\run_1$ of $\varphi_1$
and a realizer $\run_2$ of $\varphi_2$,
and that realizers of existential formulae $\ex X \varphi(X)$
are pairs $\pair{\seq,\run}$ consisting of a witness $\seq$ for the $\exists X$
and a realizer $\run$
of $\varphi(\seq)$.

Our starting point is the fact that $\MSO$
on $\omega$-words can be completely axiomatized
as a subsystem of second-order Peano arithmetic~\cite{siefkes70lnm}
(see also~\cite{riba12ifip}).
From the classical axiomatization of $\MSO$,
we derive an intuitionistic
variant $\SMSO$
(for \emph{Synchronous} $\MSO$).
$\SMSO$ comes equipped with an extraction procedure 
which is sound and complete \wrt\@ 
Church's synthesis:
proofs in $\SMSO$ of formulae of the form $\ex {\vec Y} \varphi(\vec X;\vec Y)$
(with only $\vec X$ free)
can be translated to DMMs
and such proofs exist for all solvable instances of Church's synthesis.
Our approach is \emph{Safraless} in the sense that while
we do rely on McNaughton's Theorem for the \emph{correctness}
of the extraction procedure
(\ie\@ for the ``\emph{Adequacy Lemma}'' of realizability),
we never have to actually use McNaughton's Theorem
when extracting DMMs from $\SMSO$-proofs
(so that the extracted DMMs never involve determinization of automata on $\omega$-words).%
\footnote{On the other hand, usual \emph{Safraless} approaches to synthesis use
McNaughton Theorem essentially to bound the search space for potential finite-state
realizers, see \eg~\cite{kv05focs,kpv06cav,fjr11fmsd}.}

The paper is organized as follows.
We first recall in~\S\ref{sec:prelim} some background on $\MSO$ and Church's synthesis.
Our intuitionistic system $\SMSO$ is then presented in~\S\ref{sec:smso}.
We provide in \S\ref{sec:synch} some technical material
as well as detailed examples on the representation
of
DMMs
in $\MSO$,
and~\S\ref{sec:real} presents our realizability model.
Finally, in~\S\ref{sec:fib} we rephrase the realizability model
in terms of \emph{indexed categories} (see \eg~\cite{jacobs01book}),
an essential step for further generalizations.

We also have included three appendices.
They give detailed arguments and constructions that we
wished not
to put in the body of the paper, either because they are
necessary but unsurprising technicalities
(App.\ \ref{sec:app:prelim} and~\ref{sec:app:real:atom}),
or because they concern important but side results,
proved with different techniques than those emphasized in this paper
(App.\ \ref{sec:app:synch:synchbounded}).

\section{Church's Synthesis and \texorpdfstring{$\MSO$}{MSO} on Infinite Words}%
\label{sec:prelim}

\subsection{Notations}%
\label{sec:prelim:not}
Alphabets (denoted $\Sigma,\Gamma,\etc$) are finite non-empty sets.
Concatenation of words $s,t$ is denoted $s.t$,
and $\es$ is the empty word.
We use the vectorial notation both for words and finite sequences,
so that
\eg\@ $\vec\seq$ denotes a finite sequence $\seq_1,\dots,\seq_n$
and $\vec{\al a}$ denotes a word $\al a_1.\cdots.\al a_n \in \Seq\Sigma$. 
Given an $\omega$-word (or stream) $\seq \in \Sigma^\omega$ and $n \in \NN$
we write $\seq\restr n$ for the finite word $\seq(0).\cdots.\seq(n-1) \in \Seq\Sigma$. 
For each $k \in \NN$,
we still write $k$ for the function from
$\NN$ to $\two = \{0,1\}$ which takes $n$ to $1$ iff $n = k$.

\subsection{Church's Synthesis and Synchronous Functions}
Church's synthesis consists in the automatic extraction
of deterministic letter-to-letter
stream transducers (or \emph{deterministic Mealy machines})
from input-output specifications
(see \eg~\cite{thomas08npgi}).

\begin{exa}%
\label{ex:prelim:spec}
As a typical specification,
consider,
for a machine which outputs streams $\seqbis \in \two^\omega$
from input streams $\seq \in \two^\omega$,
the 
behavior
(from~\cite{thomas08npgi})
expressed by
\[
\Phi(\seq,\seqbis)
\quad\stackrel{\text{def.}}\iff\quad
\left\{
\begin{array}{l !{\quad\imp\quad} l !{\qquad} l}
  \forall n ( \seq(n) = 1
& \seqbis(n)=1 )
& \text{and}
\\
  \forall n ( \seqbis(n) = 0
& \seqbis(n+1)= 1 )
& \text{and}
\\
\multicolumn{3}{l}{
  ( \exists^\infty n~ \seq(n) = 0 ) \quad\imp\quad
  ( \exists^\infty n~ \seqbis(n) = 0 ) }
\end{array}
\right.
\]

\noindent
In words, the relation
$\Phi(\seq,\seqbis)$ imposes
$\seqbis(n) \in \two$ to be $1$ whenever $\seq(n) \in \two$ is $1$,
$\seqbis$ not be $0$ at two consecutive positions,
and moreover $\seqbis$ to be infinitely often $0$
whenever $\seq$ is infinitely often $0$.
\end{exa}

We are interested in the realization of such specifications
by finite-state deterministic letter-to-letter stream transducers or
(deterministic) Mealy machines.

\begin{defi}[Deterministic Mealy Machine]%
\label{def:mealy}
A \emph{deterministic Mealy machine} (DMM) $\At M$
with input alphabet $\Sigma$ and output alphabet $\Gamma$
(notation $\At M : \Sigma \to \Gamma$)
is given by a finite set of states $Q_{\At M}$
with a distinguished initial state $\init q_{\At M} \in Q_{\At M}$,
and a transition function
$\trans_{\At M} : Q_{\At M} \times \Sigma \to Q_{\At M} \times \Gamma$.

We write
$\trans_{\At M}^o$ for
$\pi_2 \comp \trans_{\At M} : Q_{\At M} \times\Sigma \to \Gamma$
and $\Seq\trans_{\At M}$ for the map $\Seq\Sigma \to Q_{\At M}$
obtained by iterating $\trans_{\At M}$ from the initial state:
$\Seq\trans_{\At M}(\es) \deq \init q_{\At M}$
and
$\Seq\trans_{\At M}(\vec{\al a}.{\al a}) \deq
\pi_1(\trans_{\At M}(\Seq\trans_{\At M}(\vec{\al a}),\al a))$.
\end{defi}

A DMM $\At M : \Sigma \to \Gamma$ induces
a function $F : \Sigma^\omega \to \Gamma^\omega$
obtained by iterating $\trans_{\At M}^o$ along the input:
$F(\seq)(n) = \trans_{\At M}^o(\Seq\trans_{\At M}(\seq\restr n),\seq(n))$.
Hence $F$ can produce a prefix of length $n$ of its output from a
prefix of length $n$ of its input.
These functions are called \emph{synchronous} (or \emph{causal}).

\begin{defi}[Synchronous Function]
A function $F : \Sigma^\omega \to \Gamma^\omega$ is
\emph{synchronous} if for all $n \in \NN$
and all $\seq,\seqbis \in \Sigma^\omega$
we have $F(\seq)\restr n = F(\seqbis) \restr n$
whenever $\seq \restr n = \seqbis \restr n$.
We say that a synchronous function $F$ is \emph{finite-state}
if it is induced by a DMM\@.

We write $F : \Sigma \to_{\Synch} \Gamma$ when $F$ is a synchronous
function $\Sigma^\omega \to \Gamma^\omega$,
and $F : \Sigma \to_{\Mealy} \Gamma$ when $F$ is finite-state
synchronous.
\end{defi}

\begin{exas}%
\label{ex:prelim:mealy}
\hfill
\begin{enumerate}
\item\label{ex:prelim:mealy:id}
The identity function $\Sigma^\omega \to \Sigma^\omega$
is finite-state synchronous as being
induced by the DMM with state set $\one = \{\unit\}$
and identity transition function
$\trans : (\unit,\al a) \longmapsto (\unit, \al a)$.

\item\label{ex:prelim:mealy:succ}
The DMM depicted in Fig.\ \ref{fig:prelim:mealy}
(left)
induces a synchronous function $F : \two^\omega \to \two^\omega$
such that $F(\seq)(n+1) = 1$ iff $\seq(n) = 1$.

\item\label{ex:prelim:mealy:spec}
The DMM depicted in Fig.\ \ref{fig:prelim:mealy} (right),
taken from~\cite{thomas08npgi},
induces a synchronous function
$F : \two \to_\Mealy \two$
such that $\Phi(\seq,F(\seq))$ holds for all $\seq \in \two^\omega$,
where $\Phi$ is the relation of Ex.~\ref{ex:prelim:spec}.

\item\label{ex:prelim:mealy:pred}
Synchronous functions are obviously continuous
(taking the product topology on $\Sigma^\omega$ and $\Gamma^\omega$,
with $\Sigma,\Gamma$ discrete), but there are continuous functions
which are not synchronous,
for instance the function $P : \two^\omega \to \two^\omega$
such that $P(\seq)(n) = 1$ iff $\seq(n+1) = 1$.
\end{enumerate}
\end{exas}

\begin{figure}[t!]
\begin{center}
\begin{tabular}{c !{\qquad\quad} c} 
%

\begin{tikzpicture}[->,node distance=3.5cm,every state/.style={inner sep=2}]
  \node[state,initial,initial above,initial text=] (A) {$0$} ;
  \node[state] (B) [right of=A] {$1$};

\path (A) edge [loop left] node[left, style={font=\footnotesize}]
  {$0|0$} (A) ;

\path
  (A) edge [bend left] node[below, style={font=\footnotesize}] {$1|0$} (B) ;

\path (B) edge [bend left] node[above,style={font=\footnotesize}] {$0|1$} (A) ;

\path (B) edge [loop right] node[right, style={font=\footnotesize}]
  {$1|1$} (B) ;
\end{tikzpicture}

&

\begin{tikzpicture}[->,node distance=3.5cm]
  \node[state,initial,initial above,initial text=] (A) {$0$} ;
  \node[state] (B) [right of=A] {$1$};

\path
  (A) edge [bend left] node[below, style={font=\footnotesize}] {$0|1 \,,\, 1|1$} (B) ;

\path (B) edge [loop right] node[right,style={font=\footnotesize}] {$1|1$} (A) ;
\path (B) edge [bend left] node[above,style={font=\footnotesize}] {$0|0$} (A) ;
\end{tikzpicture}
\end{tabular}
\end{center}
\caption{Examples of DMMs
(where a transition $\al a|\al b$ outputs $\al b$ from input $\al a$).%
\label{fig:prelim:mealy}}
\end{figure}

\noindent
For the definition and adequacy of our realizability interpretation
(\S\ref{sec:real}),
it is convenient to note that alphabets and (finite-state)
synchronous functions form a category.

\begin{prop}%
\label{prop:prelim:synch:cat}
Synchronous functions form a category $\Synch$ whose objects
are alphabets and
whose morphisms from $\Sigma$ to $\Gamma$ are the synchronous functions
$\Sigma \to_\Synch \Gamma$.
The identity on $\Sigma$ is the synchronous function of
Ex.\ \ref{ex:prelim:mealy}.(\ref{ex:prelim:mealy:id}), 
and composition is usual function composition.
Moreover, if the $\Synch$-maps
$F : \Sigma \to_\Synch \Gamma$
and
$G : \Gamma \to_\Synch \Delta$
are finite-state,
then so is $G \comp F$.
\end{prop}

\begin{proof}
Since sets and functions form a category, $\Synch$ is a category
as soon as composition of functions preserves synchronicity.
Consider synchronous $G : \Gamma^\omega \to \Delta^\omega$
and $F : \Sigma^\omega \to \Gamma^\omega$.
Assume $\seq,\seqbis \in \Sigma^\omega$ and $n \in \NN$
such that $\seq\restr n = \seqbis \restr n$.
Then since
$F$ is synchronous it follows that $F(\seq)\restr n = F(\seqbis)\restr n$,
and since $G$ is synchronous we deduce
$G(F(\seq))\restr n = G(F(\seqbis))\restr n$,
that is
\[
(G \comp F)(\seq)\restr n ~~=~~ (G \comp F)(\seqbis)\restr n
\]


For the second part of the statement,
assume that $G$ and $F$ are induced respectively by
$\At N : \Gamma \to \Delta$ and $\At M : \Sigma \to \Gamma$.
Then $G \comp F$ is induced by the DMM
\[
(\At K : \Sigma \to \Delta) \quad\deq\quad
\big(
Q_{\At M} \times Q_{\At N} ~,~
(\init q_{\At M} , \init q_{\At N}) ~,~
\trans_{\At K}
\big)
\]

\noindent
whose transition function
\[
  \trans_{\At K}
\quad:\quad
(Q_{\At M} \times Q_{\At N}) \times \Sigma
\quad\longto\quad
(Q_{\At M} \times Q_{\At N}) \times \Delta
\]

\noindent
takes
$((q_{\At M},q_{\At N}) \,,\, \al a)$
to
$((q'_{\At M},q'_{\At N}) \,,\, \al d)$
with
\begin{align*}
  (q'_{\At N} \,,\, \al d)
&\hspace{2ex}:=\hspace{2ex} \trans_{\At N}(q_{\At N},\al b) \\[-1mm]
  (q'_{\At M} \,,\, \al b)
&\hspace{2ex}:=\hspace{2ex} \trans_{\At M}(q_{\At M},\al a)
    \tag*{\qedhere}
\end{align*}
\end{proof}

Proposition~\ref{prop:prelim:synch:cat}
implies that $\Synch$ has a wide subcategory consisting
of \emph{finite-state} functions.

\begin{defi}[The Category $\Mealy$]%
\label{def:prelim:synchcat}
Let $\Mealy$ be the category whose objects are alphabets
and whose morphisms from $\Sigma$ to $\Gamma$
are finite-state synchronous \emph{functions}
$\Sigma \to_\Mealy \Gamma$.
\end{defi}

Note that for $\Mealy$ to be a category
(namely for the associativity and identity laws of composition)
it is essential that $\Mealy$-maps consist of \emph{functions}
rather than \emph{machines}.
The following obvious fact
is useful for our realizability model (\S\ref{sec:real}).

\begin{rem}%
\label{rem:prelim:lift}
Functions $\al f : \Sigma \to \Gamma$ induce
$\Mealy$-maps $\lift{\al f} : \Sigma \to_\Mealy \Gamma$.
\end{rem}

It is also worth noticing that the category
$\Mealy$ has finite products.

\begin{prop}%
\label{prop:prelim:synchcart}
The category $\Mealy$ has finite products.
The product of $\Sigma_1,\dots,\Sigma_n$ (for $n \geq 0$)
is given by the $\Set$-product
$\Sigma_1 \times \cdots \times \Sigma_n$
(so that $\one$ is terminal in $\Mealy$).
\end{prop}



\begin{figure}[t!]
\[
\begin{array}{c}
\textbf{Atoms:}\qquad
\alpha
\quad\bnf\quad x \Eq y \gs x \Leq y \gs \Succ(x,y) \gs \Zero(x) \gs x \In X 
  \gs \True \gs \False
\\[0.5em]
\begin{array}{r !{\quad} r !{\quad\bnf\quad} l}
  \textbf{Deterministic formulae:}
& \delta,\delta'
& \alpha \gs \delta \land \delta' \gs \lnot \varphi
\\[0.5em]
  \textbf{$\MSO$ formulae:}
& \varphi,\psi
& \delta
  \gs \varphi \land \psi
  \gs \ex x \varphi \gs \ex X \varphi
\end{array}
\end{array}
\]
\caption{The Formulae of $\MSO$ and $\SMSO$.%
\label{fig:mso:form}}
\end{figure}

\subsection{Monadic Second-Order Logic (\texorpdfstring{$\MSO$}{MSO}) on Infinite Words}
We consider a formulation of 
$\MSO$ based on a purely relational two-sorted language,
with a specific choice of atomic formulae.
There is a sort of \emph{individuals}, with variables $x,y,z,\etc$,
and a sort of \emph{(monadic) predicates}, with variables $X,Y,Z,\etc$
Our formulae for $\MSO$, denoted $\varphi,\psi,\etc$,
are given in Fig.\ \ref{fig:mso:form}.
They are defined by mutual induction with the \emph{deterministic formulae}
(denoted $\delta,\delta',\etc$) from atomic formulae ranged over by $\alpha$.

$\MSO$ formulae are interpreted in the standard model $\Std$ of $\omega$-words
as usual.
Individual variables range over natural numbers $n,m,\ldots \in \NN$
and predicate variables range over sets of natural numbers
$\seq,\seqbis,\ldots \in \Po(\NN) \iso \two^\omega$.
The atomic predicates are interpreted as expected: $\Eq$ is equality,
$\In$ is membership,
$\Leq$ is the relation $\leq$ on $\NN$, $\Succ$ is the successor relation,
and $\Zero$ holds on $n$ iff $n=0$.
We write $\Std \models \varphi$ when the closed formula $\varphi$
holds under this interpretation.

We often write $X(x)$ or even $X x$ for $x \In X$.
We also use the following abbreviations.

\begin{nota}%
\label{not:prelim:der}
Given formulae $\varphi$ and $\psi$, we let
\[
\begin{array}{r c l !{\qquad\quad} r c l}
  \varphi \lor \psi
& \deq
& \lnot(\lnot \varphi \land \lnot \psi)

& \all x \varphi
& \deq
& \lnot \ex x \lnot \varphi
\\

  \varphi \limp \psi
& \deq
& \lnot(\varphi \land \lnot \psi)

& \all X \varphi
& \deq
& \lnot \ex X \lnot \varphi
\\

  \varphi \liff \psi
& \deq
& (\varphi \limp \psi) \land (\psi \limp \varphi)

& x \Lt y
& \deq
& \big( x \Leq y ~\land~ \lnot(x \Eq y) \big)
\\

  \exists {x \Leq y}~ \varphi
& \deq
& \ex x(x \Leq y ~\land~ \varphi)

& \forall {x \Leq y}~ \varphi
& \deq
& \all x(x \Leq y ~\limp~ \varphi)
\\

  \exists {x \Lt y}~ \varphi
& \deq
& \ex x(x \Lt y ~\land~ \varphi)

& \forall {x \Lt y}~ \varphi
& \deq
& \all x(x \Lt y ~\limp~ \varphi)
\end{array}
\]

\noindent
We moreover let, for $z$ not free in $\varphi$:
\[
\begin{array}{r c l !{\qquad\quad} r c l}
  \exists^\infty x~ \varphi
& \deq
& \all z \ex x (z \Leq x ~\land~ \varphi)

& \forall^\infty x~ \varphi
& \deq
& \ex z \all x(z \Leq x ~\limp~ \varphi)
\end{array}
\]
\end{nota}

\noindent
$\MSO$ on $\omega$-words is known to be decidable by
Büchi's Theorem~\cite{buchi62lmps}.

\begin{thm}[Büchi~\cite{buchi62lmps}]%
\label{thm:mso:dec}
$\MSO$ over $\Std$ is decidable.
\end{thm}

\noindent
Following~\cite{buchi62lmps} (but see also \eg~\cite{pp04book,vw08chapter}),
the (non-deterministic) automata method for deciding $\MSO$
proceeds by a recursive translation of $\MSO$-formulae
to \emph{non-deterministic Büchi automata} (NBAs).
An NBA is an NFA running on $\omega$-words and for which
an (infinite) run is accepting if it has infinitely many occurrences
of final states.


The crux of Büchi's Theorem is the effective
closure of
NBAs
under complement.
Let us recall a few known facts
on the complementation of
NBAs (see \eg~\cite{thomas97handbook,gtw02alig}).
First,
the translation of $\MSO$-formulae to automata is non-elementary.
Second, its is known that \emph{deterministic} Büchi automata (DBAs) are strictly
less expressive than NBAs. 
Finally, it is known that complementation of
NBAs
is algorithmically hard:
there is a family of languages ${(\Lang_n)}_{n >0}$
such that each $\Lang_n$ can be recognized by
an NBA
with $n+2$ states,
but such that the complement of $\Lang_n$ cannot be recognized
by
an NBA
with
less than $n!$ states. 

\subsection{Church's Synthesis for \texorpdfstring{$\MSO$}{MSO}}
Church's synthesis problem for $\MSO$ is the following.
Given as input an $\MSO$-formula $\varphi(\vec X;\vec Y)$
(with only free variables $\vec X = X_1,\dots,X_p$ and $\vec Y = Y_1,\dots,Y_q$),
\begin{enumerate}
\item\label{item:church:decide}
decide whether there exist finite-state synchronous functions
$\vec F = F_1,\dots,F_q$
(with each $F_i : \two^p \to_\Mealy \two$,
so that $\vec F : \two^p \to_\Mealy \two^q$ collectively)
such that $\Std \models \varphi(\vec\seq;\vec F(\vec\seq))$
for all $\vec\seq \in {(\two^\omega)}^p \iso {(\two^p)}^\omega$, and

\item\label{item:church:dmm}
construct such $\vec F$ whenever they exist.
\end{enumerate}

Given a formula $\varphi(\vec X,\vec Y)$
as above,
we say that $\vec F : \two^p \to_\Mealy \two^q$ realizes
$\varphi(\vec X;\vec Y)$ in the sense of Church
(or \emph{Church-realizes} $\varphi(\vec X;\vec Y)$)
when
$\varphi(\vec\seq,\vec F(\vec\seq))$ holds for all $\vec\seq$.

\begin{exa}%
\label{ex:mso:spec}
The specification $\Phi$ of Ex.\ \ref{ex:prelim:spec}
can be officially written in the language of $\MSO$
as the following formula $\phi(X;Y)$:
\[
\all t \big( X t ~\limp~ Y t \big)
~~\land~~
\all t \forall {t'}
\big( \Succ(t,t') ~\limp~ \lnot Y t ~\limp~ Y t' \big)
~~\land~~
\big( (\exists^\infty t~ \lnot X t) ~\limp~ (\exists^\infty t~ \lnot Y t) \big)
\]
\end{exa}

\noindent
Church's synthesis has been shown to be solvable by Büchi \& Landweber~\cite{bl69tams},
using automata on $\omega$-words and infinite two-player games
(a solution is also possible via tree automata~\cite{rabin72}):
there is an algorithm which, on input $\varphi(\vec X;\vec Y)$,
(\ref{item:church:decide})
decides whether
a finite-state
synchronous Church-realizer of $\varphi(\vec X;\vec Y)$ exists,
and if yes (\ref{item:church:dmm}) 
provides
a DMM 
implementing it.

The standard algorithm solving Church's synthesis for $\MSO$
(see \eg~\cite{thomas08npgi})
proceeds via
McNaughton's Theorem
(\cite{mcnaughton66ic}, see also \eg~\cite{pp04book,thomas97handbook}),
which states that Büchi automata can be translated to equivalent \emph{deterministic}
finite-state automata,
but equipped
with stronger acceptance conditions than Büchi automata.
There are different variants of such conditions:
\emph{Muller}, \emph{Rabin}, \emph{Streett} or \emph{parity} conditions
(see \eg~\cite{thomas97handbook,gtw02alig,pp04book}).
All of them can specify states which \emph{must not}
occur infinitely often in an accepting run.
For the purpose of this paper, we only need to consider the simplest of them,
the Muller conditions.
A \emph{Muller condition} is given by a family of set of states $\MF$,
and a run is accepting when the set of states occurring infinitely often in
it belongs to the family $\MF$.


\begin{thm}[McNaughton~\cite{mcnaughton66ic}]%
\label{thm:mso:mcnaughton}
Each
NBA
is equivalent to a deterministic Muller
automaton.
\end{thm}

\noindent
There is a lower bound in $2^{\Omega(n \log(n))}$
on the number of states
of a deterministic Muller automaton equivalent to
an NBA
with $n$ states~\cite{yan08lmcs}.
The best known constructions for McNaughton's Theorem
(such as \emph{Safra's construction} or its variants)
give deterministic Muller automata with $2^{O(n\log(n))}$ states from
NBAs
with $n$ states.



The standard solution
to Church's synthesis for $\MSO$
starts by translating $\varphi(\vec X;\vec Y)$ to a deterministic
Muller automaton,
and then turns this deterministic automaton into a two-player sequential game,
in which the Opponent $\forall$bélard plays input
bit sequences
in $\two^p$ while the Proponent $\exists$loïse replies
with output
bit sequences
in $\two^q$,
so that Proponent's strategies correspond to synchronous functions
$\two^p \to_{\Synch} \two^q$.
The game is equipped with an $\omega$-regular winning condition
(induced by the acceptance condition of the Muller automaton).
The solution is then provided by Büchi-Landweber Theorem~\cite{bl69tams},
which states that $\omega$-regular games on finite graphs are effectively
determined, and moreover that the winner always has a finite-state winning strategy.

\begin{exa}
Consider the last conjunct
$\phi_2[X,Y] \deq (\exists^\infty t~ \lnot X t) \limp (\exists^\infty t~ \lnot Y t)$
of the formula $\phi(X;Y)$ of Ex.\ \ref{ex:mso:spec}.
When translating $\phi_2$ to a finite-state automaton,
the positive occurrence of $(\exists^\infty t~ \lnot Y t)$ can be translated
to a DBA\@. 
However, the negative occurrence of $(\exists^\infty t~ \lnot X t)$
corresponds to $(\forall^\infty t~ X t)$
and cannot be translated to a
\emph{deterministic} Büchi automaton.
Even if  
a very simple two-states Muller automaton exists
for $(\forall^\infty t~ X t)$,
McNaughton's Theorem~\ref{thm:mso:mcnaughton}
is in general required for
Boolean combinations of $\exists^\infty t\,(-)$'s.
\end{exa}

\begin{figure}
\[
\begin{array}{c}
\dfrac{}{\vec\varphi,\varphi \thesis \varphi}
\qquad\quad

\dfrac{\vec\varphi \thesis \psi \qquad \vec\varphi,\psi\thesis \varphi}
  {\vec\varphi \thesis \varphi}
\qquad\quad

\dfrac{\vec\varphi \thesis \varphi \qquad \vec\varphi \thesis \lnot \varphi}
  {\vec\varphi \thesis \bot}
\qquad\quad

\dfrac{\vec\varphi,\varphi \thesis \bot}{\vec\varphi \thesis \lnot \varphi}
\\\\

\dfrac{\vec\varphi \thesis \varphi \qquad \vec\varphi \thesis \psi}
  {\vec\varphi \thesis \varphi \land \psi}
\qquad\quad
\dfrac{\vec\varphi \thesis \varphi \land \psi}{\vec\varphi \thesis \varphi}
\qquad\quad
\dfrac{\vec\varphi \thesis \varphi \land \psi}{\vec\varphi \thesis \psi}
\qquad\quad
\dfrac{\vec\varphi \thesis \varphi[y/x]}{\vec\varphi \thesis \ex x \varphi}
\qquad\quad
\dfrac{\vec\varphi \thesis \varphi[Y/X]}{\vec\varphi \thesis \ex X \varphi}
\\\\
\dfrac{\vec\varphi \thesis \ex x \varphi
  \qquad
  \vec\varphi, \varphi \thesis \psi}
  {\vec\varphi \thesis \psi}
~\text{($x$ not free in $\vec\varphi,\psi$)}

\hfill
\dfrac{\vec\varphi \thesis \ex X \varphi
  \qquad
  \vec\varphi, \varphi \thesis \psi}
  {\vec\varphi \thesis \psi}
~\text{($X$ not free in $\vec\varphi,\psi$)}
\end{array}
\]
\caption{Logical Rules of $\MSO$ and $\SMSO$.
\label{fig:mso:ded}}
\end{figure}

\begin{figure}
\[
\begin{array}{c}
\dfrac{\vec\varphi \thesis \varphi}
  {\vec\varphi,\psi \thesis \varphi}
\qquad\quad

\dfrac{\vec\varphi \thesis \bot}{\vec\varphi \thesis \varphi}
\qquad\quad

\dfrac{\vec\varphi,\varphi \thesis \psi}
  {\vec\varphi \thesis \varphi \limp \psi}
\qquad\quad

\dfrac{\vec\varphi \thesis \varphi \limp \psi
  \qquad
  \vec\varphi \thesis \varphi}
  {\vec\varphi \thesis \psi}
\qquad\quad

\\\\
\dfrac
  {\vec\varphi \thesis \varphi}
  {\vec\varphi \thesis \varphi \lor \psi}
\qquad\quad

\dfrac
  {\vec\varphi \thesis \psi}
  {\vec\varphi \thesis \varphi \lor \psi}
\qquad\quad

\dfrac{\vec\varphi \thesis \varphi \lor \psi
  \qquad
  \vec\varphi,\varphi \thesis \theta
  \qquad
  \vec\varphi,\psi \thesis \theta}
  {\vec\varphi \thesis \theta}

\\\\

\dfrac{\vec\varphi \thesis \varphi}{\vec\varphi \thesis \all x \varphi}
~\text{\small ($x$ not free in $\vec\varphi$)}

\qquad

\dfrac{\vec\varphi \thesis \all x \varphi}
  {\vec\varphi \thesis \psi[y/x]}

\qquad

\dfrac{\vec\varphi \thesis \varphi}{\vec\varphi \thesis \all X \varphi}
~\text{\small ($X$ not free in $\vec\varphi$)}

\qquad

\dfrac{\vec\varphi \thesis \all X \varphi}
  {\vec\varphi \thesis \psi[Y/X]}

\end{array}
\]
\caption{Admissible Rules of $\MSO$.}\label{fig:mso:der}
\end{figure}

\begin{figure}
\begin{description}
\setlength{\itemsep}{1em}
\item[Equality Rules]
\[
\dfrac{}{\vec\varphi \thesis x \Eq x}
\qquad\qquad
\dfrac{\vec\varphi \thesis \varphi[y/x] \qquad \vec\varphi \thesis y \Eq z}
  {\vec\varphi \thesis \varphi[z/x]}
\]

\item[Partial Order Rules]
\[
\dfrac{}{\vec\varphi \thesis x \Leq x}
\qquad\qquad
\dfrac{\vec\varphi \thesis x \Leq y \qquad \vec\varphi \thesis y \Leq z}
  {\vec\varphi \thesis x \Leq z}
\qquad\qquad
\dfrac{\vec\varphi \thesis x \Leq y \qquad \vec\varphi \thesis y \Leq x}
  {\vec\varphi \thesis x \Eq y}
\]

\item[Basic $\Zero$ and $\Succ$ Rules (total injective relations)]
\[
\begin{array}{c}
\dfrac{}{\vec\varphi \thesis \ex y \Zero(y)}

\qquad\qquad

\dfrac{\vec\varphi  \thesis \Zero(x)
  \qquad \vec\varphi \thesis \Zero(y)}
  {\vec\varphi \thesis x \Eq y}

\\\\


\dfrac{}{\vec\varphi \thesis \ex y \Succ(x,y)}

\qquad\qquad

\dfrac{\vec\varphi \thesis \Succ(y,x)
  \qquad
  \vec\varphi \thesis \Succ(z, x)}
  {\vec\varphi \thesis y \Eq z}

\qquad\quad

\dfrac{\vec\varphi \thesis \Succ(x,y)
  \qquad \vec\varphi \thesis \Succ(x,z)}
  {\vec\varphi \thesis y \Eq z}
\end{array}
\]

\item[Arithmetic Rules]
\[
\dfrac{\vec\varphi \thesis \Succ(x,y)
  \quad~~
  \vec\varphi \thesis \Zero(y)}
  {\vec\varphi \thesis \False}
\qquad~~
\dfrac{\vec\varphi \thesis \Succ(x, y)}{\vec\varphi \thesis x \Leq y}
\qquad~~
\dfrac{\vec\varphi \thesis \Succ(y,y')
  \quad~~
  \vec\varphi\thesis x \Leq y'
  \quad~~
  \vec\varphi \thesis\lnot(x \Eq y')}
  {\vec\varphi\thesis x \Leq y}
\]
\end{description}
\caption{Arithmetic Rules of $\MSO$ and $\SMSO$.%
\label{fig:mso:arith}}
\end{figure}

\subsection{An Axiomatization of \texorpdfstring{$\MSO$}{MSO}}%
\label{sec:prelim:ax}
Our approach to Church's synthesis relies on the fact that the $\MSO$-theory
of $\Std$ can be completely axiomatized as a subsystem of second-order
Peano arithmetic~\cite{siefkes70lnm} (see also~\cite{riba12ifip}).
For the purpose of this paper, it is convenient
to axiomatize $\MSO$ with
the non-logical rules of Fig.\ \ref{fig:mso:arith}
together with
the following
\emph{comprehension} and \emph{induction} rules:
\begin{equation}
\label{eq:mso:ax}
\dfrac{\vec\varphi \thesis \varphi[\psi[y]/X]}
   {\vec\varphi \thesis \ex X \varphi}
\qquad\qquad
\dfrac{\vec\varphi,\Zero(z) \thesis \varphi[z/x]
\qquad \vec\varphi,\Succ(y,z),\varphi[y/x] \thesis \varphi[z/x]}
   {\vec\varphi \thesis \varphi}
\end{equation}
where $z$ and $y$ do not occur free in $\vec\varphi,\varphi$,
and where $\varphi[\psi[y]/X]$
is the usual formula substitution, 
which commutes over all connectives
(avoiding the capture of free variables),
and with $(x \In X)[\psi[y]/X] = \psi[x/y]$.
As for the logical rules of $\MSO$,
we consider the presentation of two-sorted classical logic
consisting of
the rules
of Fig.\ \ref{fig:mso:ded} together with the following rule
of \emph{double negation elimination}:
\begin{equation}
\label{eq:mso:dne}
\dfrac{\vec\varphi \thesis \lnot\lnot\varphi}
  {\vec\varphi \thesis \varphi}
\end{equation}

\begin{defi}[Deduction for $\MSO$]
The deduction system of $\MSO$ is given by the rules of
two-sorted classical logic
(Fig.\ \ref{fig:mso:ded} and~(\ref{eq:mso:dne}))
together with the rules of
Fig.\ \ref{fig:mso:arith} and~(\ref{eq:mso:ax}).
\end{defi}

\noindent
We write $\vec\varphi \thesis_\MSO \varphi$ if $\vec\varphi \thesis \varphi$
is provable in $\MSO$. We also write $\MSO \thesis \varphi$ for $\thesis_\MSO\varphi$.

\begin{rem}%
\label{rem:mso:der}
As usual with classical logic, the rules of Fig.\ \ref{fig:mso:der}
(where $\limp,\lor,\forall$ are the defined connectives of
Notation~\ref{not:prelim:der})
are admissible in $\MSO$.
\end{rem}

As announced, deduction for $\MSO$ 
is complete \wrt\@ the standard model $\Std$.

\begin{thm}[Siefkes~\cite{siefkes70lnm}]%
\label{thm:mso:compl}
For every closed formula $\varphi$, we have $\Std \models \varphi$
if and only if
$\MSO \thesis \varphi$.
\end{thm}

Actually obtaining Thm.\ \ref{thm:mso:compl}
from~\cite{siefkes70lnm} or~\cite{riba12ifip} requires some easy but tedious work.
We discuss here the latter option.
The difference between~\cite{riba12ifip} and
the present system is that
the axiomatization of~\cite{riba12ifip} is expressed in terms
of the strict part of $\Leq$
(written $\Lt$, see Notation~\ref{not:prelim:der})
and that comprehension is formulated
with the following usual axiom scheme (where $X$ is not free in $\varphi$):
\begin{equation}
\label{eq:mso:ca}
\ex X \all x \big( X(x) ~~\longliff~~ \varphi[x/y] \big)
\end{equation}

We state here the properties required to bridge the gap between~\cite{riba12ifip}
and the present axiomatization of $\MSO$.
Missing details are provided in App.\ \ref{sec:app:prelim}.
First,
the comprehension scheme of the present version of $\MSO$ directly
implies~(\ref{eq:mso:ca}),
since using
\[
\forall x \big( \varphi[x/y] ~~\longliff~~ \varphi[x/y] \big)
\qquad=\qquad
\forall x \big( X(x) ~~\longliff~~ \varphi[x/y] \big)
  [\varphi[y]/X]
\]
we have
\[
\text{
  \AXC{}
  \UIC{$\vec\varphi \thesis \forall x \big( \varphi[x/y] ~~\longliff~~ \varphi[x/y] \big)$}
  \UIC{$\vec\varphi \thesis \exists X\forall x \big( X(x) ~~\longliff~~ \varphi[x/y] \big)$}
\DisplayProof}
\]

\noindent
In order to deal with the $\Lt$-axioms of~\cite{riba12ifip},
we rely on a series of arithmetical lemmas of $\MSO$
displayed in Fig.\ \ref{fig:prelim:mso:lem}.

\begin{lem}%
\label{lem:prelim:mso:lem}
$\MSO$ proves all the sequents of Fig.\ \ref{fig:prelim:mso:lem}.
\end{lem}

\begin{figure}[t!]
\[
\begin{array}{r !{~~} l}
(1) &
\thesis \lnot(x \Lt x)
\\[0.5em]

(2) &
x \Lt y,\, y \Lt z \thesis x \Lt z
\\[0.5em]

(3) &
\Succ(x,y),\, x \Eq y \thesis \False
\\[0.5em]

(4) &
\thesis \forall x\exists y (x\Lt y)
\\[0.5em]

(5) &
\Succ(y,y'),\, x \Leq y,\, x \Eq y' \thesis \False
\\[0.5em]

(6) &
\Zero(x) \thesis x \Leq y
\\[0.5em]

(7) &
x \Leq y,\, \Zero(y) \thesis \Zero(x)
\\[0.5em]

(8) &
\forall y(x \Leq y) \thesis \Zero(x)
\\[0.5em]

(9) &
x \Lt y,\, \Succ(x,x') \thesis x' \Leq y
\\[0.5em]

(10) &
x \Leq y,\, \Succ(x,x'),\, \Succ(y,y') \thesis x' \Leq y'
\\[0.5em]

(11) &
\thesis
\forall x \forall y
\Big[
y \Lt x ~~\longliff~~ \exists z \big( y \Leq z ~\land~ \Succ(z,x) \big)
\Big]
\\[0.5em]

(12) &
\thesis x \Lt y \lor x \Eq y \lor y \Lt x
\\[0.5em]

(13) &
\thesis
\forall x \forall y
\Big[
\Succ(x,y)
~~\longliff~~
\big( x \Lt y ~\land~
\lnot \exists z(x \Lt z \Lt y) \big)
\Big]
\end{array}
\]
\caption{Some Arithmetic Lemmas of $\MSO$\label{fig:prelim:mso:lem}.}
\end{figure}

\noindent
Finally, the induction axiom of~\cite{riba12ifip} is
the usual \emph{strong induction} axiom:
\begin{equation}
\label{eq:mso:ind}
\forall X \Big[
\forall x \big( \forall y(y \Lt x \limp X y) ~\longlimp~ X x \big)
~~\longlimp~~ \all x X x
\Big]
\end{equation}

\begin{lem}%
\label{lem:prelim:mso:ind}
$\MSO$ proves the strong induction axiom~(\ref{eq:mso:ind}).
\end{lem}

\noindent
The detailed proofs of Lem.\ \ref{lem:prelim:mso:lem} and
Lem.\ \ref{lem:prelim:mso:ind}
are deferred to App.\ \ref{sec:app:prelim}.

\section{\texorpdfstring{$\SMSO$}{SMSO}: A Synchronous Intuitionistic Variant of \texorpdfstring{$\MSO$}{MSO}}%
\label{sec:smso}

\noindent
We now introduce $\SMSO$, an intuitionistic variant of $\MSO$
equipped with an extraction procedure,
which is sound and complete \wrt\@ Church's synthesis:
proofs of existential statements can be
translated to
finite-state synchronous Church-realizers,
and such proofs exist for each solvable instance of Church's synthesis
(Thm.\ \ref{thm:smso:main}, \S\ref{sec:smso:main}).

As it is common with intuitionistic versions of classical systems,
$\SMSO$ has the same language as $\MSO$, and its deduction rules
are based on intuitionistic predicate calculus (Fig.\ \ref{fig:mso:ded}).
As expected,
$\SMSO$ contains $\MSO$
via negative translation.
Actually, our limited vocabulary without primitive universal quantifications
allows for a Glivenko Theorem,
in the sense that
$\SMSO$ proves $\lnot\lnot\varphi$ iff $\MSO$ proves~$\varphi$
(Thm.~\ref{thm:smso:glivenko}, \S\ref{sec:smso:glivenko}).
%
%
%
In order for $\SMSO$ to contain a negative
translation of $\MSO$ while admitting a computational
interpretation in the sense of~\S\ref{sec:real},
one has to devise appropriate counterparts
to the comprehension and induction rules of $\MSO$ (\ref{eq:mso:ax}):
\[
\dfrac{\vec\varphi \thesis \varphi[\psi[y]/X]}
   {\vec\varphi \thesis \ex X \varphi}
\qquad\qquad
\dfrac{\vec\varphi,\Zero(z) \thesis \varphi[z/x]
\qquad \vec\varphi,\Succ(y,z),\varphi[y/x] \thesis \varphi[z/x]}
   {\vec\varphi \thesis \varphi}
\]
(where $z$, $y$ do not occur free in $\vec\varphi,\varphi$).
%
First, $\SMSO$ cannot have the comprehension rule of $\MSO$.
The reason is that monadic variables
are computational objects in the realizability interpretation of $\SMSO$
(\S\ref{sec:real}),
while
the comprehension rule of $\MSO$ has instances
in which the existential monadic quantification
cannot be witnessed by computable functions from the parameters of
$\vec\varphi$, $\psi$ and $\varphi$.
The situation is similar to that of
higher-type intuitionistic (Heyting) arithmetic,
in which predicates, represented as characteristic functions,
are computational objects
(see \eg~\cite{kohlenbach08book}).\footnote{This contrasts with second-order logic
based on Girard's System F~\cite{girard72phd}
(see also~\cite{glt89book}),
in which second-order variables have no computational content.}
The usual solution in that setting
is to only admit negative translations of comprehension.
We take a similar approach for $\SMSO$.
In view of Glivenko's Theorem~\ref{thm:smso:glivenko},
this amounts to equip $\SMSO$ with the
\emph{negative comprehension} rule:
\begin{equation}
\label{eq:smso:negca}
\dfrac{\vec\varphi \thesis \varphi[\psi[y]/X]}
   {\vec\varphi \thesis \lnot \lnot \exists X \varphi}
\end{equation}

Second, for the extraction of \emph{finite-state synchronous} functions from proofs,
the induction scheme of $\MSO$ also has to be restricted.
Recall the \emph{deterministic formulae} of Fig.\ \ref{fig:mso:form}:
\[
  \delta,\delta'
\quad\bnf\quad
  \alpha \gs \delta \land \delta' \gs \lnot \varphi
\]

\noindent
Deterministic formulae are to be interpreted by deterministic
(not nec.\@ Büchi) automata, and thus
have trivial realizers in the sense of~\S\ref{sec:real}.
%
As a consequence, we can trivially realize
the following \emph{deterministic induction}
rule
(where $z, y$ do not occur free in $\vec\varphi,\delta$):
\begin{equation}
\label{eq:smso:ind}
\dfrac{\vec\varphi, \Zero(z) \thesis \delta[z/x]
   \qquad \vec\varphi, \Succ(y,z), \delta[y/x]\thesis \delta[z/x]}
   {\vec\varphi \thesis \delta}
\end{equation}

\noindent
In addition, since deterministic formulae have trivial realizers,
we can safely assume in $\SMSO$
the elimination of double negation on deterministic
formulae:
\begin{equation}
\label{eq:smso:dnedet}
\dfrac{\vec\varphi \thesis \lnot\lnot \delta}
  {\vec\varphi \thesis \delta}
\end{equation}

\noindent
Note that~(\ref{eq:smso:dnedet}) would follow,
using the rules of Fig.\ \ref{fig:mso:ded}, by simply assuming elimination of double
negation for atomic formulae.
Note also that~(\ref{eq:smso:dnedet})
would follow from induction in a setting
like Heyting arithmetic.

Furthermore, $\SMSO$ is equipped with a positive \emph{synchronous}
restriction of the comprehension rule of $\MSO$,
which gives Church-realizers for all solvable instances
of Church's synthesis.
This synchronous restriction of comprehension asks
the comprehension formula to be
\emph{uniformly bounded} in the following sense.

\begin{defi}[Relativized and Bounded Formulae]%
\label{def:smso:bound}
\hfill
\begin{enumerate}
\item Given formulae $\varphi$ and $\theta$ and a variable $y$,
the \emph{relativization of $\varphi$ to $\theta[y]$}
(notation $\varphi\restr\theta[y]$) is defined by induction
on $\varphi$ as usual:
\[
\begin{array}{c}
\alpha\restr\theta[y]
~\deq~
\alpha
\qquad
(\varphi \land \psi)\restr\theta[y]
~\deq~
\varphi\restr\theta[y] \land \psi\restr\theta[y]
\qquad
(\lnot \varphi)\restr\theta[y]
~\deq~
\lnot (\varphi\restr\theta[y])
\\[0.5em]

(\ex X \varphi)\restr\theta[y]
~\deq~
\ex X \varphi\restr\theta[y]
\qquad
(\ex x \varphi)\restr\theta[y]
~\deq~
\ex x (\theta[x/y] \land \varphi\restr\theta[y])
\end{array}
\]
where, in the clauses for $\exists$, the variables $x$ and $X$ are assumed not
to occur free in $\theta$.

\item
A formula $\hat\varphi$
is \emph{bounded by $x$}
if it is of the form $\psi\restr(y \Leq x)[y]$
(notation $\psi\restr[-\Leq x]$).
It is \emph{uniformly bounded}
if moreover $x$ is the only free individual variable of $\hat\varphi$.
\end{enumerate}
\end{defi}

\noindent
As we shall see in~\S\ref{sec:synch:char},
bounded formulae
correspond to the formulae of $\MSO$ over \emph{finite} words.
We are now ready to define the system $\SMSO$.

\begin{defi}[The Logic $\SMSO$]
The logic $\SMSO$ has the same language as $\MSO$.
Its deduction rules are those given in Fig.\ \ref{fig:mso:ded}
together with the rules of Fig.\ \ref{fig:mso:arith},
the rules (\ref{eq:smso:negca}), (\ref{eq:smso:ind}),
(\ref{eq:smso:dnedet}),
and
the following rule
of \emph{synchronous comprehension}
in which $\hat\varphi$ is uniformly bounded by $y$:
\[
\dfrac{\vec\varphi \thesis \psi[\hat\varphi[y]/X]}
  {\vec\varphi \thesis \ex X \psi}
\]
\end{defi}

\noindent
Similarly as with $\MSO$, we write $\vec\varphi \thesis_\SMSO \varphi$
if $\vec\varphi \thesis \varphi$ is provable in $\SMSO$,
and we write $\SMSO \thesis \varphi$ for $\thesis_\SMSO \varphi$.

\begin{rem}
As usual with natural deduction systems,
$\SMSO$ satisfies the substitution lemma,
which gives the admissibility of the cut rule.
We included that rule in $\SMSO$ because
it corresponds to the composition of realizers in the realizability model,
and thus
has a natural computational interpretation.
%
\end{rem}

\begin{nota}
In the following, we use a double dashed horizontal line
to denote admissible rules.
For instance, we freely use the \emph{weakening} rule
\[
\dfrac{\vec\varphi \thesis \varphi}
  {\vec\varphi,\psi \thesis \varphi}
\]
with the notation
\[
\text{
\AXC{$\vec\varphi \thesis \varphi$}
\doubleLine\dashedLine
\UIC{$\vec\varphi,\psi \thesis \varphi$}
\DisplayProof}
\]
\end{nota}

\begin{rem}
Note that $\SMSO$ has a limited set of connectives.
In contrast with $\MSO$, which is based on classical logic,
the derived connectives of Notation~\ref{not:prelim:der}
do not define the usual corresponding intuitionistic connectives.
For example,
with $\psi \limp \varphi = \lnot (\psi \land \lnot \varphi)$
as in Not.\ \ref{not:prelim:der},
while the usual $\limp$-introduction rule is admissible in $\SMSO$:
\[
\text{
\AXC{}
\UIC{$\vec\varphi,\psi\land \lnot\varphi \thesis \psi \land \lnot \varphi$}
\UIC{$\vec\varphi,\psi\land \lnot\varphi \thesis \psi$}
\AXC{$\vec\varphi,\psi \thesis \varphi$}
\doubleLine\dashedLine
\UIC{$\vec\varphi,\psi\land \lnot\varphi,\psi \thesis \varphi$}
\BIC{$\vec\varphi,\psi \land \lnot\varphi \thesis \varphi$}
\AXC{}
\UIC{$\vec\varphi,\psi \land \lnot\varphi \thesis \psi \land \lnot\varphi$}
\UIC{$\vec\varphi,\psi \land \lnot\varphi \thesis \lnot\varphi$}
\BIC{$\vec\varphi, \psi \land \lnot \varphi \thesis \False$}
\UIC{$\vec\varphi \thesis \psi \limp \varphi$}
\DisplayProof}
\]

\noindent
the elimination rule of $\limp$ is only admissible
for implications with \emph{deterministic} r.-h.s:
\[
\text{
\AXC{$\vec\varphi \thesis \psi$}
\doubleLine\dashedLine
\UIC{$\vec\varphi,\lnot\delta \thesis \psi$}
\AXC{}
\UIC{$\vec\varphi,\lnot \delta \thesis \lnot\delta$}
\BIC{$\vec\varphi,\lnot\delta \thesis \psi \land \lnot\delta$}
\AXC{$\vec\varphi \thesis \psi \limp \delta$}
\doubleLine\dashedLine
\UIC{$\vec\varphi,\lnot\delta \thesis \psi \limp \delta$}
\BIC{$\vec\varphi, \lnot\delta \thesis \False$}
\UIC{$\vec\varphi \thesis \lnot\lnot\delta$}
\UIC{$\vec\varphi \thesis \delta$}
\DisplayProof}
\]

\noindent
On the other hand, usual $\lnot$-rules are admissible in $\SMSO$
(even without using deterministic double negation elimination):
\begin{equation}
\label{eq:smso:neg}
\begin{array}{c}
\dfrac{\vec\varphi,\varphi \thesis \lnot\psi}
  {\vec\varphi,\psi \thesis \lnot\varphi}

\qquad\qquad

\dfrac{}{\vec\varphi, \varphi \thesis \lnot\lnot\varphi}

\qquad\qquad

\dfrac{\vec\varphi,\varphi \thesis \psi}
  {\vec\varphi,\lnot\psi \thesis \lnot\varphi}

\\\\

\dfrac{}{\vec\varphi,\lnot\lnot\lnot\varphi \thesis \lnot\varphi}

\qquad\qquad

\dfrac{\vec\varphi,\varphi \thesis \psi}
  {\vec\varphi,\lnot\lnot\varphi \thesis \lnot\lnot\psi}

\qquad\qquad

\dfrac{\vec\varphi,\psi \thesis \lnot\varphi}
  {\vec\varphi,\lnot\lnot\psi \thesis \lnot\varphi}
\end{array}
\end{equation}

\noindent
Indeed, the second rule of the first line follows from the first one,
and the third rule is obtained from the first two ones.
The rules of the second line all follow from the last two rules of the first
line.
Finally,
the first rule of the first line is obtained as usual:
\[
\text{
\AXC{$\vec\varphi,\varphi \thesis \lnot\psi$}
\doubleLine\dashedLine
\UIC{$\vec\varphi,\psi,\varphi \thesis \lnot\psi$}
\AXC{}
\UIC{$\vec\varphi,\psi,\varphi \thesis \psi$}
\BIC{$\vec\varphi,\psi,\varphi \thesis \False$}
\UIC{$\vec\varphi,\psi \thesis \lnot\varphi$}
\DisplayProof}
\]
\end{rem}

\subsection{A Glivenko Theorem for \texorpdfstring{$\MSO$}{MSO}}%
\label{sec:smso:glivenko}
The limited vocabulary of $\MSO$ without primitive universal quantifications
allows for a Glivenko Theorem,
in the sense that $\SMSO$ proves $\lnot\lnot\varphi$ iff $\MSO$ proves~$\varphi$.
While Glivenko's Theorem is often stated for propositional logic,
it also holds in presence of existential quantifications
(see \eg~\cite[Thm.\ 59.(b), \S81]{kleene52book} or~\cite[\S 10.1]{kohlenbach08book}),
but does not extend to universal quantifications
(see \eg~\cite[\S 11]{kohlenbach08book}).
In particular, we would have relied on
usual recursive negative translations if $\MSO$ 
had primitive universal quantifications.

\begin{thm}[Glivenko's Theorem for $\MSO$ and $\SMSO$]%
\label{thm:smso:glivenko}
If $\MSO \thesis \varphi$, then $\SMSO \thesis \lnot \lnot \varphi$.
\end{thm}

\begin{proof}
By induction on $\MSO$-derivations, we show that if
$\vec\varphi \vdash \varphi$ is derivable in $\MSO$,
then
$\vec\varphi \vdash \lnot\lnot \varphi$
is derivable in $\SMSO$.
This amounts to showing that for every $\MSO$-rule of the form
\[
\dfrac{{(\vec\varphi_i \thesis \varphi_i)}_{i \in I}}
  {\vec\psi \thesis \psi}
\]
the following rule is admissible in $\SMSO$:
\[
\dfrac{{(\vec\varphi_i \thesis \lnot\lnot\varphi_i)}_{i \in I}}
  {\vec\psi \thesis \lnot\lnot\psi}
\]

\noindent
The logical rules of $\MSO$
may be treated exactly as in the usual proof of Glivenko's Theorem.
%
It remains to deal with the non-logical rules of $\MSO$.
\begin{description}
\item[Comprehension]
%
We have to prove that the following is admissible in $\SMSO$:
\[
\dfrac{\vec \varphi \thesis \lnot\lnot\psi[\varphi[y]/X]}
  {\vec\varphi \thesis \lnot\lnot \ex X \psi}
\]

\noindent
But this directly follows from the negative comprehension scheme of $\SMSO$
together with the last rule of~(\ref{eq:smso:neg}):
\[
\text{
\AXC{$\vec\varphi \thesis \lnot\lnot \psi[\varphi[y]/X]$}
\AXC{}
\UIC{$\vec\varphi, \psi[\varphi[y]/X] \thesis \psi[\varphi[y]/X]$}
\UIC{$\vec\varphi, \psi[\varphi[y]/X] \thesis \lnot\lnot \ex X \psi$}
\doubleLine\dashedLine
\UIC{$\vec\varphi, \lnot\lnot \psi[\varphi[y]/X] \thesis \lnot\lnot \ex X \psi$}
\BIC{$\vec\varphi \thesis \lnot\lnot \ex X \psi$}
\DisplayProof}
\]

\item[Induction]
We need to show that the following is admissible in $\SMSO$:
\[
\tag{$z,y$ not free in $\vec\varphi,\varphi$}
\dfrac{\vec\varphi, \Zero(z) \thesis \lnot\lnot \varphi[z/x]
  \qquad \vec\varphi, \Succ(y,z), \varphi[y/x] \thesis \lnot\lnot\varphi[z/x]}
 {\vec\varphi \thesis \lnot\lnot\varphi}
\]

\noindent
But this follows from deterministic induction together
with the last rule of~(\ref{eq:smso:neg}):
\[
\text{
\AXC{$\vec\varphi, \Zero(z) \thesis \lnot\lnot \varphi[z/x]$}
\AXC{$\vec\varphi, \Succ(y,z), \varphi[y/x] \thesis \lnot\lnot\varphi[z/x]$}
\doubleLine\dashedLine
\UIC{$\vec\varphi, \Succ(y,z), \lnot\lnot \varphi[y/x] \thesis \lnot\lnot\varphi[z/x]$}
\BIC{$\vec\varphi \thesis \lnot\lnot \varphi$}
\DisplayProof}
\]

\item[Arithmetic Rules (Fig.\ \ref{fig:mso:arith})]
All these rules can be treated the same way.
We only detail the case of elimination of equality.
We have to show that the following rule is admissible in $\SMSO$:
\begin{equation}
\label{eq:smso:eq}
\dfrac{\vec\varphi \thesis \lnot\lnot\varphi[x/z]
  \qquad \vec\varphi \thesis \lnot\lnot(x \Eq y)}
  {\vec\varphi \thesis \lnot\lnot\varphi[y/z]}
\end{equation}

\noindent
First, note that elimination of equality in $\SMSO$ gives
\[
\dfrac{}
  {\vec\varphi, \varphi[x/z], x \Eq y \thesis \varphi[y/z]}
\]
from which (\ref{eq:smso:neg}) gives
\[
\dfrac{}
  {\vec\varphi, \lnot\lnot\varphi[x/z], \lnot\lnot(x \Eq y) \thesis
  \lnot\lnot\varphi[y/z]}
\]

\noindent
We then obtain the rule (\ref{eq:smso:eq}) by successively cutting $\lnot\lnot\varphi[x/z]$
and $\lnot\lnot(x \Eq y)$
with the corresponding premise of (\ref{eq:smso:eq}).
\qedhere
\end{description}
\end{proof}

\subsection{The Main Result}%
\label{sec:smso:main}
We are now ready to state the main result of this paper,
which says that $\SMSO$ is correct and complete (\wrt\@ its provable existentials)
for Church's synthesis.

\begin{thm}[Main Theorem]%
\label{thm:smso:main}
Consider a formula $\varphi(\vec X;\vec Y)$
with only $\vec X,\vec Y$ free.
\begin{enumerate}
\item\label{thm:smso:main:cor}
From a proof of $\exists {\vec Y} \varphi(\vec X;\vec Y)$
in $\SMSO$,
one can extract a finite-state synchronous Church-realizer of
$\varphi(\vec X;\vec Y)$.

\item\label{thm:smso:main:compl}
If $\varphi(\vec X;\vec Y)$ admits a (finite-state) synchronous Church-realizer,
then
$\ex {\vec Y} \nt\varphi(\vec X;\vec Y)$
is provable in $\SMSO$.
\end{enumerate}
\end{thm}

\noindent
The correctness part~(\ref{thm:smso:main:cor}) of Thm.\ \ref{thm:smso:main}
is be proved in~\S\ref{sec:real} using a notion of realizability
for $\SMSO$ based on automata and synchronous finite-state
functions.
The completeness part~(\ref{thm:smso:main:compl}) is proved
in~\S\ref{sec:synch:repr}, relying on the completeness of the axiomatization
of $\MSO$ (Thm.\ \ref{thm:mso:compl}) together with the correctness
of the negative translation $\nt{(-)}$ (Thm.\ \ref{thm:smso:glivenko}).

\section{On the Representation of Deterministic Mealy Machines in \texorpdfstring{$\MSO$}{MSO}}%
\label{sec:synch}

\noindent
This section gathers several (possibly known)
results related to the representation of DMMs in $\MSO$.
We begin in~\S\ref{sec:synch:repr} with the completeness part of
Thm.\ \ref{thm:smso:main},
which follows usual representations of automata in $\MSO$
(see \eg~\cite[\S5.3]{thomas97handbook}).
In~\S\ref{sec:synch:rec},
we then recall from~\cite{siefkes70lnm,riba12ifip}
the \emph{Recursion Theorem}, 
which is
a convenient tool to reason on runs of deterministic automata in $\MSO$.
In~\S\ref{sec:synch:char} we state a 
Lemma for the correctness part of Thm.\ \ref{thm:smso:main},
which relies on the usual translation of $\MSO$-formulae over \emph{finite words}
to DFAs (see~\eg~\cite[\S3.1]{thomas97handbook}).
Finally, in~\S\ref{sec:synch:synchbounded} we give a possible strengthening of
the synchronous comprehension
rule of $\SMSO$, based on Büchi's Theorem~\ref{thm:mso:dec}.

We work with the following notion of representation.
Recall from~\S\ref{sec:prelim:not} that
for $k \in \NN$,
we still write $k$ for the function from
$\NN$ to $\two$ which takes $n$ to $1$ iff $n = k$.

\begin{defi}[Representation]%
\label{def:synch:repr}
Let $\varphi$ be a formula with free variables
among $z,x_1,\dots,x_\ell$, $X_1,\dots,X_p$. 
We say that $\varphi$ \emph{$z$-represents}
$F : \two^\ell \times \two^p \longto_\Mealy \two$ if
for all $n \in \NN$,
all $\vec\seq \in {(\two^\omega)}^p$,
and all $\vec k \in \NN^\ell$ such that $k_i \leq n$
for all $i \leq \ell$,
we have
\begin{equation}
\label{eq:synch:repr}
F(\vec k,\vec\seq)(n) = 1
\qquad\text{iff}\qquad
\Std \models \varphi[n/z,\vec k/\vec x,\vec\seq/\vec X]
\end{equation}
\end{defi}

\noindent
For $F : \two^\ell \times \two^p \to_\Mealy \two$
as in Def.~\ref{def:synch:repr},
we write
$F : \two^p \to_\Mealy \two$
(resp.\@ $F : \two^\ell \to_\Mealy \two$)
in case $\ell = 0$ (resp.\@ $p = 0$).

\subsection{Internalizing Deterministic Mealy Machines in \texorpdfstring{$\MSO$}{MSO}}%
\label{sec:synch:repr}
The completeness part~(\ref{thm:smso:main:compl}) of Thm.\ \ref{thm:smso:main}
relies on the following simple fact.

\begin{prop}%
\label{prop:synch:repr}
For every 
finite-state synchronous $F : \two^p \longto_\Mealy \two$,
one can build a deterministic uniformly bounded formula
$\delta(\vec X,x)$ 
which $x$-represents $F$.
\end{prop}

\begin{proof}
The proof is a simple adaptation of the usual pattern
(see \eg~\cite[\S5.3]{thomas97handbook}).
Let $F : \two^p \to_\Mealy \two$ be induced by
a DMM $\At M$.
W.l.o.g.\@ we can assume the state set of $\At M$
to be $\two^q$ for some $q \in \NN$.
The transition function $\trans$ of $\At M$ is thus of the form
\[
\trans ~~:~~
  \two^q \times \two^p
  ~~\longto~~
  \two^q \times \two
\]

\noindent
Let $\form I[s_1,\dots,s_q]$
be a propositional formula in the propositional variables
$s_1,\dots,s_q$ such that for $\vec{\al s} \in \two^q$,
$\form I[\vec{\al s}]$ holds iff $\vec{\al s}$ is the initial
state of $\At M$.
Further, let
\[
\form H[s_1,\dots,s_q \,,\, a_1,\dots,a_p \,,\, b \,,\, s'_1,\dots,s'_q]
\]
be a propositional formula in the propositional variables
$s_1,\dots,s_q, a_1,\dots,a_p, b, s'_1,\dots,s'_q$
such that for $\vec{\al s} \in \two^q$, $\vec{\al a} \in \two^p$,
$\al b \in \two$ and $\vec{\al s}' \in \two^q$,
we have
$\form H[\vec{\al s}, \vec{\al a}, \al b , \vec{\al s}']$
iff
$\trans(\vec{\al s} , \vec{\al a}) = (\vec{\al s}' , \al b)$.
Then $F$ is $x$-represented by the formula
\begin{equation}
\label{eq:synch:form:repr}
\delta(\vec X,x) ~~\deq~~
\forall {\vec Q,Y}
\left(
\left[
\begin{array}{l}
\forall t \Leq x ( \Zero(t) \limp \form I[\vec Q(t)] ) \quad\land
\\
\all {t,t' \Leq x}
  ( \Succ(t,t') \limp \form H[\vec Q(t),\vec X(t),Y(t),\vec Q(t')] )
\end{array}
\right]
\longlimp%
Y x
\right)
\end{equation}
where $\vec X = X_1,\dots,X_p$ codes sequences of inputs,
$Y$ codes sequences of outputs,
and where $\vec Q = Q_1,\dots,Q_q$
codes runs.
\end{proof}

\begin{rem}%
\label{rem:synch:repr}
In the proof of Prop.\ \ref{prop:synch:repr}, since $\At M$ is deterministic,
we can assume
the formula $\form I[\vec Q(t)]$ to be of the form
$\bigconj_{1 \leq i \leq q}[Q_i(t) \liff \form B_i]$ with 
$\form B_i \in \{\True,\False\}$,
and,
for some propositional formulae $\form O[-,-],\vec{\form D}[-,-]$,
the formula $\form H[\vec Q(t),\vec X(t),Y(t),\vec Q(t')]$
to be of the form
\[
\left( Y(t) \longliff \form O[\vec Q(t),\vec X(t)] \right)
~~\land~~
\bigconj_{1 \leq i \leq q}
\left(Q_i(t') \longliff \form D_i[\vec Q(t),\vec X(t)] \right)
\]
where $\form O$ codes the outputs of $\At M$ while the
$\form D_i$'s represent its transition relation on states.
\end{rem}

\begin{exa}%
\label{ex:synch:mealy}
The function induced by
the DMM
of Ex.\ \ref{ex:prelim:mealy}.(\ref{ex:prelim:mealy:spec}) 
(depicted in Fig.\ \ref{fig:prelim:mealy}, right),
is represented by a formula of the form~(\ref{eq:synch:form:repr})
with $\vec Q = Q$ (since the machine has state set $\two$),
$\vec X=X$,
and
where
$\form I[-] \deq [(-) \liff \False]$ (since state $0$ is initial)
and (following Rem.\ \ref{rem:synch:repr})
\begin{equation}
\label{eq:synch:thomas}
\form O[Q(t),X(t)]
\quad=\quad \form D[Q(t),X(t)]
\quad=\quad
(\lnot Q(t) ~\lor~ [Q(t) \land X(t)])
\end{equation}
\end{exa}

The completeness of our approach to Church's synthesis
is obtained as follows.
\begin{proof}[Proof of Thm.\ \ref{thm:smso:main}.(\ref{thm:smso:main:compl})] 
Assume that
$\varphi(\vec X ; \vec Y)$
admits a
realizer
$F : \two^p \longto_\Mealy \two^q$.
Using the Cartesian structure of $\Mealy$
(Prop.\ \ref{prop:prelim:synchcart}),
we write $F = \vec F = F_1,\dots,F_q$ with $F_i : \two^p \to_\Mealy \two$.
We thus have
$\Std \models \varphi[\vec\seq/\vec X, \vec F(\vec\seq)/\vec Y]$
for all $\vec\seq \in {(\two^\omega)}^p \iso {(\two^p)}^\omega$.
Now, by Prop.\ \ref{prop:synch:repr}
there are
uniformly bounded (deterministic) formulae
$\vec\delta = \delta_1,\dots,\delta_q$,
with free variables among $\vec X,x$,
and
such that~(\ref{eq:synch:repr}) holds for all $i=1,\dots,q$.
It thus follows that
$\Std \models \forall \vec X \varphi[\vec{\delta[x]}/\vec Y]$.
Then, by completeness (Thm.\ \ref{thm:mso:compl})
we know that
$\thesis \varphi[\vec\delta[x]/\vec Y]$
is provable in $\MSO$,
and by negative translation (Thm.\ \ref{thm:smso:glivenko})
we
get $\SMSO \thesis \lnot\lnot\varphi[\vec\delta[x]/\vec Y]$.
We can then apply ($q$ times) the synchronous comprehension scheme of $\SMSO$
and obtain
$\SMSO \thesis \exists {\vec Y} \lnot\lnot\varphi(\vec X ; \vec Y)$.
\end{proof}

\begin{exa}%
\label{ex:synch:spec}
Recall the specification of Ex.\ \ref{ex:prelim:spec} from~\cite{thomas08npgi},
represented in $\MSO$ by the formula $\phi(X;Y)$ of Ex.\ \ref{ex:mso:spec}.
Write $\phi(X;Y) = \phi_0(X,Y) \land \phi_1(X,Y) \land \phi_2(X,Y)$
where
\[
\begin{array}{r !{\quad\deq\quad} l}
  \phi_0(X,Y)
& \all t (X t ~\limp~ Y t)
\\
  \phi_1(X,Y)
& \all t \all{t'} \big( \Succ(t,t') ~\limp~ \lnot Y t ~\limp~ Y t' \big)
\\
  \phi_2(X,Y)
& (\exists^\infty t~ \lnot X t) ~\limp~ (\exists^\infty t~ \lnot Y t)
\end{array}
\]
Note that $\phi_0$ and $\phi_1$ are monotonic in $Y$,
while $\phi_2$ is anti-monotonic in $Y$.
The formula $\phi_0$ is trivially realized by the identity
function $\two \to_\Mealy \two$
(see Ex.\ \ref{ex:prelim:mealy}.(\ref{ex:prelim:mealy:id})), 
which is itself represented by the deterministic
uniformly bounded formula $\delta_0(X,x) \deq (x \In X)$.
%
For $\phi_1$ (which asks $Y$ not to have two consecutive occurrences of $0$),
consider
\[
\delta_1(X,x)
\quad\deq\quad
\delta_0(X,x) ~\lor~ \exists t \Leq x \big( \Succ(t,x) \land \lnot X(t) \big)
\]
We have
$\MSO \thesis \phi_0[X,\delta_1[x]/Y]$ since
$\delta_0 \thesis_\MSO \delta_1$
and moreover $\MSO \thesis \phi_1[X,\delta_1[x]/Y]$
since
\[
\Succ(t,t') \,,\,  \lnot X t \,,\, \lnot\exists u \big( \Succ(u,t) \land \lnot X u \big)
~\thesis_\MSO~
X t' ~\lor~ \exists u' \big( \Succ(u',t') \land \lnot X u' \big)
\]
\end{exa}

The case of $\phi_2$ in Ex.\ \ref{ex:synch:spec} is more complex.
The point is that $\phi_2[\delta_1[x]/Y]$ does not hold because if
$\forall^\infty t~ \lnot X t$
(that is if $X$ remains constantly $0$ from some time on),
then we have $\forall^\infty t~ \delta_1[x]$
(so that $Y$ stays constantly $1$ from some time on).
On the other hand, the machine of
Ex.\ \ref{ex:prelim:mealy}.(\ref{ex:prelim:mealy:spec}) 
involves internal states,
and can be represented using a fixpoint formula
of the form~(\ref{eq:synch:form:repr}).
Reasoning on such formulae is easier with more advanced tools on $\MSO$,
that we provide in~\S\ref{sec:synch:rec}.

\subsection{The Recursion Theorem}%
\label{sec:synch:rec}
Theorem~\ref{thm:smso:main}.(\ref{thm:smso:main:compl}) 
ensures that $\SMSO$ is able to handle all solvable instances of Church's synthesis,
but it gives no hint on how to actually produce proofs. 
When reasoning on fixpoint formulae
as those representing DMMs in Prop.\ \ref{prop:synch:repr},
a crucial role is played by the
\emph{Recursion Theorem} for $\MSO$~\cite{siefkes70lnm}
(see also~\cite{riba12ifip}).
The Recursion Theorem
makes it possible to define
predicates by well-founded induction \wrt\@ the relation
$\Lt$ (Notation~\ref{not:prelim:der}).
Given formulae
$\vec\psi = \psi_1,\dots,\psi_q$
and variables $x$ and $\vec X = X_1,\dots,X_q$,
we say that
\emph{$\vec\psi$ is $x$-recursive in $\vec X$}
when the following formula
$\Rec_{\vec X}^x(\vec\psi)$ holds:
\[
\forall z \forall \vec Z \forall \vec Z'
\left(
\bigconj_{1 \leq i \leq q} \forall y \Lt z \left(Z_i y \longliff Z'_i y \right)
~~\longlimp~~
\bigconj_{1 \leq i \leq q}
  \left(\psi_i[\vec Z/\vec X,z/x] \longliff \psi_i[\vec Z'/\vec X,z/x]\right)
\right)
\]

\noindent
(where $z,\vec Z,\vec Z'$ do not occur free in $\vec\psi$).
For
$\Vec{\psi(\vec X,x)}$ $x$-recursive in $\vec X$, the Recursion
Theorem says that, provably in $\MSO$,
there are unique $\vec X$ such that $\forall x(X_i x \longliff \psi_i(\vec X,x))$
holds for all $i = 1,\dots,q$.

\begin{thm}[Recursion Theorem~\cite{siefkes70lnm}]%
\label{thm:synch:rec}
$\MSO$ proves the following:
\[
\begin{array}{l}
\Rec_{\vec X}^x(\vec\psi),~
\bigconj_{1 \leq i \leq q}
\forall z
\left(
Z_i z \longliff
\forall \vec X
\left[
\bigconj_{1 \leq j \leq q} \forall x \Leq z(X_j x \liff \psi_j)
~~\limp~~
X_i z
\right]
\right)

~\thesis
\\

\multicolumn{1}{r}{
\bigconj_{1 \leq i \leq q}
\forall x\left(
Z_i x \longliff \psi_i[\vec Z/\vec X]
\right)}

\\\\

\Rec_{\vec X}^x(\vec\psi),~
\bigconj_{1 \leq i \leq q}
\forall x(Z_i x \liff \psi_i[\vec Z/\vec X])
\land
\forall x(Z'_i x \liff \psi_i[\vec Z'/\vec X])

~\thesis~

\bigconj_{1 \leq i \leq q}
\forall x\left( Z_i x \liff Z'_i x \right)

\end{array}
\]
\end{thm}

The following examples
give instances of application of
the Recursion Theorem 
in formal
reasoning on Mealy machines in $\MSO$.
The corresponding proofs in $\SMSO$ are then obtained by Thm.~\ref{thm:smso:glivenko}.

\begin{exas}%
\label{ex:synch:rec}
\hfill
\begin{enumerate}
\item\label{ex:synch:rec:mealy}
W.r.t.\@ the representation used in Prop.\ \ref{prop:synch:repr},
let $\theta(\vec X,\vec Q,Y,x)$ be 
\[
\forall t\Leq x \big( \Zero(t) \longlimp \form I[\vec Q(t)] \big)
~~\land~~
\all {t,t' \Leq x}
  \big( \Succ(t,t') \longlimp  \form H[\vec Q(t),\vec X(t),Y(t),\vec Q(t')] \big)
\]
so that
$\delta(\vec X,x) =
\forall \vec Q \forall Y ( \theta(\vec X,\vec Q,Y,x) \limp Y x )$.
The Recursion Theorem 
implies that, provably in $\MSO$, for all $\vec X$
there are unique predicates
$\vec Q,Y$ s.t.\@ $\all x \theta(\vec X,\vec Q,Y,x)$.

Indeed, assuming $\form I$ and $\form H$ are as in
Rem.\ \ref{rem:synch:repr},
we have that $\theta(\vec X,\vec Q,Y,x)$ is equivalent to
$\theta^o(\vec Q,\vec X,Y,x)
\land
\bigconj_{1 \leq i \leq q}
\theta_i(\vec Q,\vec X,Y,x)$,
where
\[
\begin{array}{r !{\quad} r !{\quad\deq\quad} l}
&
\theta^o(\vec X,\vec Q,Y,x)
&
\all {t \Leq x}
\big( Y(t) ~\longliff~ \form O[\vec Q(t),\vec X(t)] \big)
\\

&
\theta_i(\vec X,\vec Q,Y,x)
&
\all {t \Leq x}
\big( Q_i(t) ~\longliff~ \eta_i(\vec Q,\vec X,t) \big)
\\
\text{with}
&
\eta_i(\vec X,\vec Q,t)
&
( \Zero(t) \land \form B_i )
~\lor~
\exists u \Leq t
\big( \Succ(u,t) \land \form D_i[\vec Q(u),\vec X(u)] \big)
\end{array}
\]

\noindent
Now, apply Thm.\ \ref{thm:synch:rec}
to $\vec\eta$ (resp.\@ to $\form O[\vec Q(t),\vec X(t)]$)
which is $t$-recursive in $\vec Q$ (resp.\@ in $Y$).


\item
The machine of Ex.\ \ref{ex:prelim:mealy}.(\ref{ex:prelim:mealy:spec}) 
is represented as in item~(\ref{ex:synch:rec:mealy})
with $\form O$ and $\form D$ given by~(\ref{eq:synch:thomas})
(see Ex.\ \ref{ex:synch:mealy},
recalling that the machine as only two states).
Hence $\MSO$ proves that for all $X$ there are unique $Q$, $Y$
such that $\all x \theta(X,Q,Y,x)$.
Continuing now Ex.\ \ref{ex:synch:spec}, let
\[
\delta_2(X,x)
\quad\deq\quad
\all Q \forall Y
\big( \theta(X,Q,Y,x) ~\limp~ Y x \big)
\]

\noindent
It is not difficult to derive
$\MSO \thesis \phi_0[\delta_2[x]/Y] \land \phi_1[\delta_2[x]/Y]$.
The case of $\phi_2[\delta_2[y]/Y]$ amounts to showing~$\exists^\infty t\, (\lnot X t)
~\thesis_\MSO~
\exists^\infty t\,
\ex Q \exists Y (\theta(X,Q,Y,t) ~\land~ \lnot Y t)$.
Thanks to Thm.~\ref{thm:synch:rec}, this follows from
~$\all x \theta(X,Q,Y,x)
\,,\,
\exists^\infty t\, (\lnot X t)
~\thesis_\MSO~\exists^\infty t\, (\lnot Y t)$~which itself can be derived using induction.

%
%
\end{enumerate}
\end{exas}

\subsection{From Bounded Formulae to Deterministic Mealy Machines}%
\label{sec:synch:char}
We now turn to the extraction of finite-state synchronous functions
from bounded formulae.
This provides realizers of synchronous comprehension
for Thm.\ \ref{thm:smso:main}.(\ref{thm:smso:main:cor}). 
We rely on the standard translation of $\MSO$-formulae
\emph{over finite words} to DFAs (see~\eg~\cite[\S3.1]{thomas97handbook}).

\begin{lem}%
\label{lem:synch:char}
Let $\hat\varphi$ be a 
formula with free variables
among $z,x_1,\dots,x_\ell,X_1,\dots,X_p$,
and which is bounded by $z$.
Then $\hat\varphi$ $z$-represents
a finite-state synchronous 
$F: \two^\ell \times \two^p \to_\Mealy \two$
induced by a DMM computable from $\hat\varphi$.
\end{lem}

\begin{proof}
First, given a formula $\hat\varphi$ with free variables among
$z,x_1,\dots,x_\ell,X_1,\dots,X_p$, if $\hat\varphi$ is bounded
by $z$ then $\hat\varphi$ is of the form
$\psi\restr[- \Leq z]$, where the free variables of $\psi$
are among $z,x_1,\dots,x_\ell,X_1,\dots,X_p$.
But note that $\psi\restr[-\Leq z]$ is equivalent to the formula
$(\exists t(\form{last}(t) \land \psi[t/z]))\restr[-\Leq z]$,
where $\form{last}(t) \deq \forall x (t \Leq x \limp t \Eq x)$
and where $t$ does not occur free in $\psi$.
We can therefore assume that $\hat\varphi$ is of the form
$\psi\restr[- \Leq z]$ where $\psi$ has free variables among
$x_1,\dots,x_\ell,X_1,\dots,X_p$.

Then, for all $n \in \NN$,
all $\vec k \in \NN^\ell$ with $k_i \leq n$,
and all $\vec\seq \in {(\two^\omega)}^p$,
we have
$\Std \models \psi[\vec k/\vec x,\vec\seq/\vec X]\restr[- \Leq n]$
if and only if,
\emph{in the sense of $\MSO$ over finite words},
the formula $\psi$ holds
in the finite word $\pair{\vec k,\vec\seq\restr (n+1)}$.
Let $\At A = (Q,\init q,\trans,F)$
be a DFA recognizing the language of finite words satisfying
$\psi$~\cite[Thm.\@ 3.1]{thomas97handbook}.
Consider the DMM $\At M = (Q,\init q,\trans_{\At M})$
with $\trans_{\At M}(q,\al a) = (q',b)$ where $q' = \trans(q,\al a)$
and ($b = 1$ iff $q' \in F$), and let
$F : \two^\ell \times \two^p \longto_\Mealy \two$
be the function induced by $\At M$.
We then have
\begin{align*}
& \pair{\vec k,\vec\seq\restr(n+1)} \models \psi\restr[-\Leq n]
    \tag{in the sense of $\MSO$ over finite words} \\
& \iff \text{$\At A$ accepts the finite word $\pair{\vec k,\vec\seq\restr(n+1)}$} \\
& \iff F(\vec k,\vec\seq)(n) = 1
    \tag*{\qedhere}
\end{align*}
\end{proof}

\begin{rem}%
\label{rem:synch:nelb}
There is a well-known non-elementary lower bound
for translating $\MSO$-formulae over finite words to DFAs
(see \eg~\cite[Chap.\@ 13]{gtw02alig}).
This lower bound also applies to the DMMs
which induce synchronous functions represented by bounded formulae
in the sense of Def.\ \ref{def:synch:repr}.
Indeed, given
$F : \two^\ell \times \two^p \to_\Mealy \two$
$z$-represented by $\psi\restr[-\Leq z]$
(with $z$ not free in $\psi$),
for all $n \in \NN$, all $\vec\seq \in {(\two^\omega)}^p$
and all $\vec k \in \NN^\ell$ with $k_i \leq n$,
we have $F(\vec k,\vec\seq)(n) = 1$ if and only if
$\pair{\vec k,\vec\seq\restr(n+1)} \models \psi$
(in the sense of $\MSO$ over finite words).
It follows that if $F$ is induced by a DMM
$\At M = (Q,\init q,\trans)$, then
with the DFA
$\At A \deq (Q \times \two + \{\init q\}, \init q,\trans_{\At A}, Q \times \{1\})$
where
$\trans_{\At A}(\init q,\al a) \deq \trans(\init q,\al a)$
and
$\trans_{\At A}((q,b),\al a) \deq \trans(q,\al a)$, we have
$F(\vec k,\vec\seq)(n) = 1$
iff
$\At A$ accepts the finite word $\pair{\vec k,\vec\seq\restr(n+1)}$.
Since the size of $\At A$ is in general non-elementary in the size
of $\psi$, it follows that the size of $\At M$ is in general non-elementary
in the size of $\psi\restr[- \Leq z]$.
\end{rem}

\begin{exa}
Recall the continuous but not synchronous function $P$
of Ex.\ \ref{ex:prelim:mealy}.(\ref{ex:prelim:mealy:pred}). 
The function $P$ can be used to realize a predecessor function,
and thus is represented
(in the sense of~(\ref{eq:synch:repr}))
by a formula $\varphi(X,x)$ such that
$\Std \models \varphi(\seq,n)$
iff $n+1 \in \seq$.
Note that
$\varphi$ is not equivalent to a bounded formula,
since
by Lem.\ \ref{lem:synch:char}
bounded formulae represent synchronous functions.
\end{exa}

%
%

\subsection{Semantically Bounded Formulae}%
\label{sec:synch:synchbounded}
\noindent
The synchronous comprehension scheme of $\MSO$ is motivated by
Lem.\ \ref{lem:synch:char}, which tells
that uniformly bounded formulae induce DMMs.
Recall from Def.\ \ref{def:smso:bound}
that a uniformly bounded formula
is of the form $\psi\restr[-\leq x]$
with only $x$ as free individual variable.
Uniform boundedness is a purely syntactic restriction on comprehension,
which has the advantage of being
easy to check and conceptually simple to interpret in a proof relevant semantics.
We present here a more semantic criterion on the formulae for which
comprehension remains sound in a synchronous setting.
We call a formula
$\psi(\vec X,x)$
with only $\vec X,x$ free
\emph{semantically bounded}
if the following closed formula
$\B_{\vec X}^x(\psi(\vec X, x))$, expressing that the truth value of
$\psi(\vec X,n)$ only depends on the values of $\vec{X}$ up to $n$, holds:
\[
\forall z \forall \vec{Z} \vec{Z'}
\left(
\bigconj_{1 \leq i \leq q} \forall y \Leq z \big( Z_i y \longliff Z'_i y \big)
~~\longlimp~~
\big( \psi[\vec Z/\vec X,z/x] \longliff \psi[\vec{Z'}/\vec X,z/x] \big)
\right)
\]


We show in Thm.\ \ref{thm:synch:synchbounded} below that semantically bounded $\MSO$
formulae are equivalent to uniformly bounded formulae.
Since all uniformly bounded formulae are obviously semantically bounded, we
have a semantic characterization of the formulae available for the synchronous
comprehension scheme.

\begin{thm}%
\label{thm:synch:synchbounded}
If $\MSO \thesis \B_{\vec X}^x(\psi(\vec X,x))$
and the free variables of $\psi$ are among $x,\vec X$,
then there is a uniformly bounded formula $\hat\varphi(\vec X,x)$
which is effectively computable from $\psi$ and such that
$\MSO \thesis
  \all{\vec X}\all x \big(\psi(\vec X,x) \longliff \hat\varphi(\vec X,x)\big)
$.
\end{thm}

Note that Thm.\ \ref{thm:synch:synchbounded}
in particular applies if $\SMSO \thesis \B_{\vec X}^x(\psi(\vec X,x))$.
Moreover, if $\psi(X,x)$ is $x$-recursive in $X$
(in the sense of~\S\ref{sec:synch:rec}),
then $\B_X^x(\psi(X,x))$ holds, but not conversely.

Theorem~\ref{thm:synch:synchbounded} makes it possible to
derive realizers for additional instances of comprehension,
namely
for formulae which are semantically but not uniformly bounded.
However, the algorithm underlying Thm.\ \ref{thm:synch:synchbounded}
relies on Büchi's Theorem~\ref{thm:mso:dec}, and computing
$\hat\varphi$ from $\psi$ can become quickly prohibitively expensive.

The proof of Thm.\ \ref{thm:synch:synchbounded}
relies on the decidability of $\MSO$ and on two preliminary lemmas.
The first one
is the following usual transfer property
(see~\eg~\cite{riba12ifip}).
Given a set $A \sle \Po(\NN)$, write $\Std\restr A$
for the model defined as the standard model $\Std$,
but with individuals ranging over $A$ rather than $\NN$.

\begin{lem}[Transfer]%
\label{lem:synch:transf}
Let $\varphi$ be a formula with free variables among
$\vec x = x_1,\dots,x_\ell$ and $\vec X = X_1,\dots,X_p$.
Furthermore, let $A \in \two^\omega \iso \Po(\NN)$ be non-empty.
Then for all $a_1,\dots,a_\ell \in A$ and all $\vec\seq  \in {(\two^\omega)}^p$
we have
\[
\Std\restr A \models
\varphi[\vec a/\vec x,\Vec{\seq \cap A}/\vec X]
\qquad\iff\qquad
\Std \models
(\varphi[\vec a/\vec x,\vec\seq / \vec X])\restr[A(-)]
\]
\end{lem}

The second result is the following Splitting Lemma\ \ref{lem:synch:splitting},
reminiscent of the composition method from a technical point of view.
The point of Lem.\ \ref{lem:synch:splitting} is, given a formula
$\varphi$ and a distinguished individual variable $z$, to express
$\varphi$ using an elementary combination of formulae, each local either to
the initial segment $[- \Leq z]$ or to the final segment $[- \Gt z]$.
Its proof is deferred to App.\ \ref{sec:app:synch:synchbounded}.
Write $\FV^\indiv(\varphi)$ for the set of free individual
variables of the formula $\varphi$.

\begin{lem}[Splitting]%
\label{lem:synch:splitting}
Consider a formula $\psi$ and some individual variable $z$.
For every set of individual variables $V$ with $z \in V$,
one can produce a natural number $N$ and
two matching sequences of length $N$ of left formulae
${(L_j)}_{j < N}$
and right formulae ${(R_j)}_{j < N}$
such that the following holds:
\begin{itemize}
\item For every $j < N$,
$\FV^\indiv(L_j) \sle \FV^\indiv(\psi) \cap V$ and
$\FV^\indiv(R_j) \sle \FV^\indiv(\psi) \setminus V$.

\item
If $\FV^\indiv(\psi) = \{\vec x,z,\vec y\}$
with $V\cap\FV^\indiv(\psi) = \{\vec x,z\}$,
then
for all $n \in \NN$,
all $\vec a \leq n$ and all $\vec b > n$, we have
\[
\Std\models\quad
\psi[\vec a/\vec x,n/z,\vec b/\vec y]
~~\longliff~~
\bigdisj_{j < N}
L_j[\vec a/\vec x,n/z]\restr{[ - \Leq n ]}
~\land~
R_j[\vec b/\vec y]\restr{[ - \Gt n]}
\]
\end{itemize}
\end{lem}

We can now prove Thm.\ \ref{thm:synch:synchbounded}.

\begin{proof}[Proof of Thm.\ \ref{thm:synch:synchbounded}]
We work in the standard model $\Std$ of $\MSO$
and obtain the result by completeness (Thm.\ \ref{thm:mso:compl}).
Using Lem.\ \ref{lem:synch:splitting},
we know that $\psi(\vec{X},x)$ is equivalent to
\[
\varphi(\vec{X},x) \quad\deq\quad
\bigdisj_j L_j(x,\vec{X})\restr{[-\Leq x]} ~\land~ R_j(\vec{X})\restr{[- \Gt x]}
\]
Then, by our assumption that $\psi(\vec X,x)$
(and thus $\varphi(\vec X,x)$) is semantically bounded, we have
\[
\begin{array}{r !{\quad\longliff\quad} l}
  \varphi(\vec X,x)
& \varphi\big(\Vec{X(-) \land - \Leq x},x\big)
\\
& \bigdisj_j
  L_j \big(x, \vec{X(-) \land - \Leq x}\big) \restr{[- \Leq x]}
~\land~
  R_j \big(\vec{X(-) \land - \Leq x}\big) \restr{[- \Gt x]}
\end{array}
\]

\noindent
Again using Lem.\ \ref{lem:synch:transf},
for every $j < N$ and $n \in \NN$,
$R_j(\vec{X(-) \land - \Leq n})\restr{[- \Gt n]}$
is equivalent to $R_j(\vec{X(-) \land - \Leq n \land - \Gt n})\restr{[- \Gt n]}$.
By substitutivity, it is equivalent to $R'_j(n)\restr{[-\Gt n]}$,
where we set $R'_j \deq R_j(\vec{\bot})$.
Because $R'_j$ is closed and $\Std\restr[-\Gt n]\iso \Std$,
Lem.~\ref{lem:synch:transf}
moreover implies that
\[
\Std \models\quad \forall x \left(R'_j ~~\longliff~~ R'_j\restr[- \Gt x]\right)
\]

\noindent
Since $R'_j$ is closed, it follows from the
decidability of $\MSO$ (Thm.\ \ref{thm:mso:dec}) that we can decide whether
$\Std \models R'_j$
or
$\Std \models \lnot R'_j$. Define accordingly closed formulae $R''_j$:
\[
R''_j \quad\deq\quad \left\{
\begin{array}{ll}
\True &\text{if $\Std \models R'_j$} \\
\False &\text{if $\Std \models \lnot R'_j$} \\
\end{array} \right.
\]
Notice in particular that, contrary to the $R'_j$, the $R''_j$ are invariant under
relativization, i.e., the formulae $R''_j$ and $R''_j\restr[- \Leq x]$ are syntactically
equal.
It thus follows that our initial $\psi$ is equivalent to the following formula $\hat\varphi$,
which is effectively computable from $\psi$:
\[
\hat\varphi(\vec{X},x) \quad \deq \quad
\bigdisj_j L_j(x,\vec{X})\restr{[- \Leq x]} \land R''_j
\]
$\hat\varphi$ is uniformly bounded since
it is syntactically equal to
$\left(\bigdisj_j L_j(x,\vec{X}) \land R''_j\right)\restr{[- \Leq x]}$.
\end{proof}

\section{The Realizability Interpretation of \texorpdfstring{$\SMSO$}{SMSO}}%
\label{sec:real}

\noindent
We now present our realizability model for $\SMSO$,
and use it to prove 
Thm.\ \ref{thm:smso:main}.(\ref{thm:smso:main:cor}). 
This realizability interpretation bears some similarities
with usual realizability constructions
for the Curry-Howard correspondence
(see \eg~\cite{su06book,kohlenbach08book}).
For instance, as in the usual setting,
a realizer of a formula $\varphi_1 \land \varphi_2$
is a pair $\pair{\run_1,\run_2}$ of a realizer of $\run_1$
of $\varphi_1$ and a realizer $\run_2$ of $\varphi_2$.
Similarly, a realizer of $\ex X \varphi(X)$
is a pair $\pair{\seq,\run}$ of an $\omega$-word $\seq \in \two^\omega$
and a realizer $\run$ of $\varphi(\seq)$.
However, our construction departs from the standard one
on negation (for which we use McNaughton's Theorem~\ref{thm:mso:mcnaughton}),
and for the fact that there is no primitive notion of implication in $\SMSO$.
In particular, in contrast with the usual settings,
our notion of realizability for sequents of the form
$\psi \thesis \varphi$ (see Thm.\ \ref{thm:real:adeq} and Def.\ \ref{def:real:cat})
is not based on a notion of implication internal to the logic under consideration.

Our approach to Church's synthesis via realizability uses automata
in two different ways.
First, from a \emph{proof} $\Der D$ in $\SMSO$ of an existential
formula $\exists {\vec Y} \varphi(\vec X;\vec Y)$,
one can compute 
a finite-state synchronous Church-realizer $\vec F$ of $\varphi(\vec X;\vec Y)$.
Second, the adequacy of realizability
(and in particular the correctness of $\vec F$ \wrt\@ $\varphi(\vec X;\vec Y)$)
is \emph{proved} using automata for $\varphi(\vec X;\vec Y)$
obtained by McNaughton's Theorem,
but these automata do not have to be built during the extraction procedure.

\subsection{Uniform Automata}
The adequacy of realizability relies on the notion
of \emph{uniform automata} (adapted from~\cite{riba16dialaut}).
In our context, uniform automata are essentially usual non-deterministic
automata, but in which non-determinism is expressed via an explicitly
given set of \emph{moves}. This allows for a simple inheritance of the Cartesian
structure of synchronous functions (Prop.\ \ref{prop:prelim:synchcart}),
and thus to interpret
the strictly positive existentials of $\SMSO$ 
similarly as usual (weak) sums of type theory.
In particular, the set of moves $\Moves(\At A)$ of an automaton $\At A$
interpreting a formula $\varphi$ exhibits the strictly positive
existentials of $\varphi$ as
$\Moves(\At A) = \Moves(\varphi)$ where
\begin{equation}
\label{eq:real:moves}
\Moves(\alpha) \iso \Moves(\lnot\varphi) \iso \one
\qquad
\Moves(\varphi \land \psi) \iso \Moves(\varphi) \times \Moves(\psi)
\qquad
\Moves(\ex {(-)} \varphi) \iso \two \times \Moves(\varphi)
\end{equation}

\begin{defi}[(Non-Deterministic) Uniform Automata]%
\label{def:real:aut}
A (non-deterministic) \emph{uniform automaton} $\At A$ over 
$\Sigma$
(notation $\At A : \Sigma$)
has the form
\begin{equation}
\label{eq:aut}
\At A \quad=\quad
(Q_{\At A} \,,\, \init q_{\At A} \,,\, \Moves(\At A) 
  \,,\, \trans_{\At A} \,,\, \Omega_{\At A})
\end{equation}
where $Q_{\At A}$ is the finite set of \emph{states}, $\init q_{\At A} \in Q_{\At A}$
is the \emph{initial state},
$\Moves(\At A)$
is the finite non-empty set of \emph{moves},
the \emph{acceptance condition} $\Omega_{\At A}$
is an $\omega$-regular subset of $Q_{\At A}^\omega$,
and the \emph{transition function} $\trans_{\At A}$ has the form
\[
\trans_{\At A} \quad:\quad
Q_{\At A} \times \Sigma
\times
\Moves(\At A)
  \quad\longto\quad
Q_{\At A}
\]

A \emph{run} of $\At A$ on an $\omega$-word $\seq \in \Sigma^\omega$
is an $\omega$-word
$\run \in {\Moves(\At A)}^\omega$.
We say that $\run$ is \emph{accepting}
(notation $\run \real \At A(\seq)$)
if ${(q_k)}_{k \in \NN} \in \Omega_{\At A}$
for the sequence of states ${(q_k)}_{k \in \NN}$ defined as
$q_0 \deq \init q_{\At A}$
and
$q_{k+1} ~\deq~ \trans_{\At A}(q_k,\seq(k),\run(k))$.
We say that $\At A$ \emph{accepts} $\seq$ if there exists an accepting
run of $\At A$ on $\seq$, and we let $\Lang(\At A)$,
the \emph{language} of $\At A$,
be the set of $\omega$-words accepted by $\At A$.
\end{defi}

\noindent
Following the usual terminology, an automaton $\At A$ as in~(\ref{eq:aut})
is \emph{deterministic} if $\Moves(\At A)\iso \one$.

Let us now sketch how uniform automata are used in our realizability
interpretation of $\SMSO$.
First, by adapting to our context usual constructions on automata
(\S\ref{sec:real:op}),
to each formula $\varphi$ with
free variables among (say) $\vec X = X_1,\dots,X_p$,
we associate a uniform automaton
$\I\varphi$
over $\two^p$
(Fig.\ \ref{fig:real:form}).
Then, from an $\SMSO$-derivation $\Der D$ of
a sequent (say) $\varphi \thesis \psi$,
with free variables among $\vec X$ as above,
we extract a finite-state synchronous function
$F_{\Der D} : \two^p \times \Moves(\I\varphi) \longto_\Mealy \Moves(\I\psi)$
such that
$F_{\Der D}(\vec\seq,\run) \real \I\psi(\vec\seq)$
whenever $\run \real \I\varphi(\vec\seq)$.
In the case of
$\thesis \exists Y \phi(\vec X; Y)$,
the finite-state realizer $F_{\Der D}$ is of the form
$\pair{\seqbis,G}$ with $\seqbis$ and $G$ finite-state synchronous functions
$\seqbis : \two^p \longto_\Mealy \two$ and
$G : \two^p \longto_\Mealy \Moves(\phi)$
such that
$G(\vec\seq) \real \I\phi(\vec\seq,\seqbis(\vec\seq))$
for all $\vec\seq$.
%
This motivates the following notion.


\begin{defi}[The Category $\Aut_\Sigma$]%
\label{def:real:cat}
For each alphabet $\Sigma$, the category $\Aut_\Sigma$
has automata $\At A : \Sigma$ as objects.
Morphisms $F$ from $\At A$ to $\At B$
(notation $\At A \real F : \At B$)
are finite-state synchronous maps 
$F : \Sigma \times \Moves(\At A) \longto_\Mealy \Moves(\At B)$
such that
$F(\seq,\run) \real \At B(\seq)$
whenever $\run \real \At A(\seq)$.
\end{defi}

The identity morphism $\At A \real \Id_{\At A} : \At A$
is given by $\Id_{\At A}(\seq,\run) \deq \run$,
and the composition of morphisms $\At A \real F : \At B$
and $\At B \real G : \At C$ is the morphism
$\At A \real G \comp F : \At C$ given by
$(G \comp F)(\seq,\run) \deq G(\seq,F(\seq,\run))$.
It is easy to check the usual identity and composition laws of categories,
namely:
\[
\Id \comp F = F
\qquad
F \comp \Id = F
\qquad
(F \comp G) \comp H = F \comp (G \comp H)
\]

\begin{rems}%
\label{rem:real}
\hfill
\begin{enumerate}
\item\label{rem:real:lgge}
Note that if $\At B \real F : \At A$ for some $F$, then
$\Lang(\At B) \sle \Lang(\At A)$.

\begin{proof}
Assume $\At B \real F : \At A$ and $\seq \in \Lang(\At B)$
so that $\run \real \At B(\seq)$ for some $\run \in {\Moves(\At B)}^\omega$.
Then by definition of $\At B \real F : \At A$,
we have $F(\seq,\run) \real \At A(\seq)$ and thus $\seq \in \Lang(\At A)$.
\end{proof}

\item
One could 
also
consider the category $\AUT_\Sigma$ defined as $\Aut_\Sigma$,
but with maps not 
required to be finite-state.
All statements of \S\ref{sec:real} hold for $\AUT_\Sigma$,
but for Cor.\ \ref{cor:real:main}, which
would lead to non necessarily finite-state realizers
and 
would not give Thm.\ \ref{thm:smso:main}.(\ref{thm:smso:main:cor}). 

\item
Uniform automata are a variation of usual automata on $\omega$-words,
which is convenient for our purposes,
namely the adequacy of our realizability interpretation.
Hence,
while it would have been possible to define uniform automata with any
of the usual acceptance conditions (see \eg~\cite{thomas97handbook}),
we lose nothing by assuming their
acceptance conditions to be given by arbitrary $\omega$-regular sets.

\item
Given automata $\At A,\At B : \Sigma$, checking the existence of a realizer
$\At A \real F : \At B$ can be reduced
(\eg\@ using the tools of~\cite{pr18lics,riba16dialaut})
to checking the existence of a winning strategy for the Proponent
($\exists$loïse) in an $\omega$-regular game on a finite graph,
which can in turn be decided by the Büchi-Landweber Theorem~\cite{bl69tams}.
\end{enumerate}
\end{rems}

\subsection{Constructions on Automata}%
\label{sec:real:op}

\noindent
We gather here constructions on uniform automata that we need to interpret formulae.
First, automata are closed under the following operation of
\emph{finite substitution}.

\begin{prop}%
\label{prop:real:subst}
Given $\At A : \Sigma$ and a function $\al f : \Gamma \to \Sigma$,
let $\At A\lift{\al f} : \Gamma$ be the automaton identical to $\At A$,
but with
$\trans_{\At A\lift{\al f}}(q,\al b,u) \deq \trans_{\At A}(q,\al f(\al b),u)$.
Then $\seq \in \Lang(\At A\lift{\al f})$ iff
$\al f \comp \seq \in \Lang(\At A)$.
\end{prop}

\begin{exa}%
\label{ex:real:subst}
Assume $\At A$ 
interprets
a formula
$\varphi$ with free variables among $\vec X$, 
so that
$\vec\seq \in \Lang(\At A)$ iff $\Std \models \varphi[\vec\seq/\vec X]$.
Then $\varphi$ is also a formula with free variables among
$\vec X,\vec Y$, and we have
$\vec\seq\vec{\seq'} \in \Lang(\At A\lift\pi)$ iff
$\Std\models\varphi[\vec\seq/\vec X/\vec{\seq'}/\vec Y]$,
where $\pi : \vec X \times \vec Y \to \vec X$ is a projection.
\end{exa}

\noindent
The Cartesian structure of $\Mealy$
lifts to $\Aut_\Sigma$.
This gives the interpretation of conjunctions.

\begin{prop}%
\label{prop:real:cart}
For each $\Sigma$, the category $\Aut_\Sigma$ has finite products.
Its terminal object is the automaton $\tu = (\one,\unit,\one,\trans_\tu,\one^\omega)$,
where $\trans_\tu(-,-,-) = \unit$.
%
Binary products are given
by
\[
\begin{array}{r !{\quad} r !{\quad\deq\quad} l}
& \At A \times \At B
& (Q_{\At A} \times Q_{\At B} \,,\,
  (\init q_{\At A},\init q_{\At B}) \,,\,
  \Moves(\At A) \times \Moves(\At B) \,,\, \trans \,,\, \Omega)
\\
  \text{where}
& \trans((q_{\At A},q_{\At B}),\, \al a,\, (u,v))
& (\trans_{\At A}(q_{\At A},\al a,u) ~,~ \trans_{\At B}(q_{\At B},\al a,v))
\end{array}
\]

\noindent
and where ${(q_n,q'_n)}_n \in \Omega$ iff
(${(q_n)}_n \in \Omega_{\At A}$ and ${(q'_n)}_n \in \Omega_{\At B}$).
Note that $\Omega$ is $\omega$-regular since
$\Omega_{\At A}$ and $\Omega_{\At B}$ are $\omega$-regular.
Moreover, $\Lang(\tu) = \Sigma^\omega$ and
$\Lang(\At A \times \At B) = \Lang(\At A) \cap \Lang(\At B)$.
\end{prop}

\begin{proof}
The Cartesian structure is directly inherited from $\Mealy$ and is omitted.
Moreover, we obviously have $\Lang(\tu) = \Sigma^\omega$.
Let us show that $\Lang(\At A_1 \times \At A_2) = \Lang(\At A_1) \cap \Lang(\At A_2)$.
The inclusion $(\sle)$ follows from
Rem.\ \ref{rem:real}.(\ref{rem:real:lgge}) 
applied to the projection maps
$\At A_1 \times \At A_2 \real \varpi_i : \At A_i$
induced by the Cartesian structure.
For the converse inclusion $(\sge)$,
note that if $\run_i \real \At A_i(\seq)$
for $i =1,2$,
then $\pair{\run_1,\run_2} \real (\At A_1 \times \At A_2)(\seq)$.
\end{proof}

Uniform automata are equipped with the obvious adaptation
of the usual projection on non-deterministic automata,
which interprets existentials.
Given a uniform automaton $\At A : \Sigma \times \Gamma$,
its \emph{projection on $\Sigma$} is the automaton
\[
(\exists_\Gamma \At A : \Sigma)
~\deq~
(Q_{\At A} \,,\, \init q_{\At A} \,,\, \Gamma \times \Moves(\At A)
\,,\, \trans \,,\, \Omega_{\At A})
\quad\text{where}\quad
\trans(q,\al a,(\al b,u)) ~\deq~ \trans_{\At A}(q,(\al a,\al b),u)
\]

\begin{prop}%
\label{prop:real:realproj}
Given $\At A : \Sigma \times \Gamma$ and $\At B : \Sigma$,
the realizers $\At B \real F : \exists_\Gamma \At A$
are exactly
the $\Mealy$-pairs $\pair{\seqbis,G}$ of finite-state synchronous
functions
\[
\seqbis ~~:~~
\Sigma \times \Moves(\At B) ~~\longto_\Mealy~~ \Gamma
\qquad\qquad
G ~~:~~
\Sigma \times \Moves(\At B) ~~\longto_\Mealy~~ \Moves(\At A)
\]
such that
$G(\seq,\run) \real \At A\pair{\seq,\seqbis(\seq,\run)}$
for all $\seq \in \Sigma^\omega$ and all $\run \real \At B(\seq)$.
\end{prop}

\begin{proof}
Consider a realizer $\At B \real F : \exists_\Gamma \At A$
for some $\At B : \Sigma$.
Then $F$ is a finite-state synchronous function from
$\Sigma^\omega \times {\Moves(\At B)}^\omega$ to
${(\Gamma \times \Moves(\At A))}^\omega \iso \Gamma^\omega \times {\Moves(\At A)}^\omega$,
and is therefore given by a pair $\pair{\seqbis,G}$ of finite-state synchronous
functions
\begin{equation}
\label{eq:real:realproj}
\seqbis ~~:~~
\Sigma \times \Moves(\At B) ~~\longto_\Mealy~~ \Gamma
\qquad\qquad
G ~~:~~
\Sigma \times \Moves(\At B) ~~\longto_\Mealy~~ \Moves(\At A)
\end{equation}

\noindent
Moreover, given $\seq \in \Sigma^\omega$ and $\run \real \At B(\seq)$,
since $F(\seq,\run) \real \exists_{\Gamma}\At A(\seq)$,
it is easy to see that
$G(\seq,\run) \real \At A(\pair{\seq,\seqbis(\seq,\run)})$.
Conversely, given $\seqbis$ and $G$ as in~(\ref{eq:real:realproj}),
if
$G(\seq,\run) \real \At A(\pair{\seq,\seqbis(\seq,\run)})$
for all $\seq \in \Sigma^\omega$ and all $\run \real \At B(\seq)$,
then we have $\At B \real \pair{\seqbis,G} : \exists_\Gamma \At A$.
\end{proof}

The negation $\lnot(-)$ on formulae is interpreted by an
operation $\aneg(-)$ on uniform automata which involves
McNaughton's Theorem~\ref{thm:mso:mcnaughton}.

\begin{prop}%
\label{prop:real:compl}
Given a uniform automaton $\At A : \Sigma$,
there is a uniform \emph{deterministic} automaton $\aneg \At A : \Sigma$
such that $\seq \in \Lang(\aneg \At A)$ iff
$\seq \notin \Lang(\At A)$.
\end{prop}

\begin{proof}
Let $U \deq \Moves(\At A)$
and consider the (usual) deterministic automaton
$\At S$ over
$\Sigma \times U$
with the same states as $\At A$ and with transition function
$\trans_{\At S}$
defined as
$\trans_{\At S}(q,(\al a,u)) \deq \trans_{\At A}(q,\al a,u)$.
Then $\run \real \At A(\seq)$ iff $\At S$ accepts $\pair{\seq,\run}$.
Since $\Omega_{\At A}$ is $\omega$-regular, 
it is recognized by a non-deterministic Büchi automaton
$\At C$ over $Q_{\At A}$.
We then obtain a non-deterministic Büchi automaton $\At B$
over $\Sigma \times U$ with state set $Q_{\At A} \times Q_{\At C}$
and
s.t.\@ $\Lang(\At B) = \Lang(\At S)$.
It follows that $\seq \in \Lang(\At A)$ iff
$\seq \in \Lang(\weak\exists_{U} \At B)$,
where $\weak\exists_{U} \At B$ is the usual projection of $\At B$ on $\Sigma$
(see \eg~\cite{thomas97handbook}).
By McNaughton's Theorem~\ref{thm:mso:mcnaughton},
$\weak\exists_{U} \At B$ is equivalent to a \emph{deterministic}
Muller automaton $\At D$ over $\Sigma$.
Then we let $\aneg \At A$ 
be the deterministic uniform automaton defined as $\At D$ but with
$\Omega_{\aneg \At A}$ the $\omega$-regular set
generated by the Muller condition
$S \in \MF$ iff $S \notin \MF_{\At D}$
(see~\eg~\cite[Thm.\@ I.7.1 \& Prop.~I.7.4]{pp04book}).
\end{proof}

\subsection{The Realizability Interpretation}%
\label{sec:real:real}
We are now going to define our realizability interpretation.
This goes in two steps:
\begin{enumerate}
\item\label{item:real:real:form}
To each formula $\varphi$ we associate a uniform automaton $\I\varphi$.

\item\label{item:real:real:der}
To each derivation $\Der D$ of a (closed) sequent
$\varphi_1,\dots,\varphi_n \thesis \varphi$
in $\SMSO$, we associate a finite-state synchronous $F_{\Der D}$
such that
$\I{\varphi_1} \times \cdots \times \I{\varphi_n} \real F_{\Der D} : \I\varphi$.
\end{enumerate}

\noindent
We first discuss step~(\ref{item:real:real:form}).
Consider a formula $\varphi$ with free variables among
$\vec x = x_1,\dots,x_\ell$
and
$\vec X = X_1,\dots,X_p$.
Its interpretation is a uniform automaton
$\I\varphi_{\vec x,\vec X}$ over $\two^\ell \times \two^p$,
defined by induction on $\varphi$,
and such that $\I\delta_{\vec x,\vec X}$ is deterministic for a deterministic $\delta$.
We thus have to devise a deterministic uniform automaton $\At A(\alpha)$
for each atomic formula $\alpha$ of $\SMSO$.
The definitions of the $\At A(\alpha)$'s
are easy and follow usual constructions
(see \eg~\cite{thomas97handbook}). They are deferred to App.\ \ref{sec:app:real:atom}.
Moreover, in order to handle individual variables,
the interpretation also uses a deterministic uniform automaton $\Sing : \two$
accepting the language of $\omega$-words
$\seq \in \two^\omega \iso \Po(\NN)$ such that $\seq$ is a singleton.
App.\ \ref{sec:app:real:atom} also presents a possible definition for $\Sing$.
The interpretation
$\I\varphi_{\vec x,\vec X}$
is defined in Fig.\ \ref{fig:real:form},
where $\pi$, $\pi'$ are suitable projections
and $\sigma$ is a suitable permutation.
We write $\I\varphi$ when $\vec x,\vec X$ are irrelevant or
understood from the context.
Note that the set of moves $\Moves(\varphi)$ of $\I\varphi$
indeed satisfies (\ref{eq:real:moves}), so in particular $\I\delta$
is indeed deterministic for a deterministic $\delta$.

As expected, the interpretation $\I{-}$ is correct in the following sense.
For $k \in \NN$, we keep on writing $k$ for the function from
$\NN$ to $\two$ which takes $n$ to $1$ iff $n = k$.

\begin{prop}%
\label{prop:real:cor}
Given a formula $\varphi$ with free variables among
$\vec x = x_1,\dots,x_\ell$ and $\vec X = X_1,\dots,X_p$,
for all
$\vec k \in \NN^\ell$
and all $\vec\seq \in {(\two^\omega)}^p \iso {(\two^p)}^\omega$
we have
$(\vec k,\vec\seq) \in \Lang(\I\varphi_{\vec x,\vec X})$
iff
$\Std \models \varphi[\vec k/\vec x,\vec \seq/\vec X]$.
\end{prop}

\begin{figure}
\[
\begin{array}{c}
  \I\alpha_{\vec x,\vec X}
~\deq~
  \At A(\alpha)\lift\pi
\qquad\quad
  \I{\lnot \psi}_{\vec x,\vec X}
~\deq~
  \aneg \I\psi_{\vec x,\vec X}
\qquad\quad
  \I{\ex X \psi}_{\vec x,\vec X}
~\deq~
  \exists_\two(\I\psi_{\vec x,\vec X,X})
\\[0.5em]
  \I{\psi_1 \land \psi_2}_{\vec x,\vec X}
~\deq~
  \I{\psi_1}_{\vec x,\vec X} \times \I{\psi_2}_{\vec x,\vec X}
\qquad\qquad
  \I{\ex x \psi}_{\vec x,\vec X}
~\deq~
  \exists_\two(\Sing\lift{\pi'} \times \I\psi_{\vec x,x,\vec X}\lift{\sigma})
\end{array}
\]
\caption{Interpretation of $\SMSO$-Formulae as Uniform Automata.%
\label{fig:real:form}}
\end{figure}

We now turn to step~(\ref{item:real:real:der}).
Let $\varphi_1,\dots,\varphi_n,\varphi$ be formulae
and consider variables
$\vec x = x_1,\dots,x_\ell$ and $\vec X= X_1,\dots,X_p$
containing all the free variables of
$\varphi_1,\dots,\varphi_n,\varphi$.
Then we say that a synchronous function
\[
F \quad:\quad
\two^\ell \times \two^p
\times \one^\ell \times \Moves(\varphi_1) \times \cdots \times \Moves(\varphi_n)
\quad\longto_\Mealy\quad \Moves(\varphi)
\]
\emph{$\vec x,\vec X$-realizes} the sequent $\varphi_1,\dots,\varphi_n \thesis \varphi$
(notation $\varphi_1,\dots,\varphi_n \real_{\vec x,\vec X} F : \varphi$
or $\vec\varphi \real_{\vec x,\vec X} F : \varphi$)
if
\[
\Sing^\ell\lift{\vec\pi} \times
\I{\varphi_1}_{\vec x,\vec X} \times \cdots \times \I{\varphi_n}_{\vec x,\vec X}
~~\real~~ F ~~:~~ \I\varphi_{\vec x,\vec X}
\]
where $\vec\pi$ are suitable projections.

\begin{thm}[Adequacy]%
\label{thm:real:adeq}
Let $\vec\varphi,\varphi$ 
be formulae
with variables among $\vec x,\vec X$.
From an $\SMSO$-derivation $\Der D$
of $\vec\varphi \thesis \varphi$,
one can compute an $\Mealy$-morphism $F_{\Der D}$
s.t.\@
$\vec\varphi \real_{\vec x,\vec X} F_{\Der D} : \varphi$.
\end{thm}

\noindent
Adequacy of realizability,
together with Prop.\ \ref{prop:real:realproj},
directly gives
Theorem~\ref{thm:smso:main}.(\ref{thm:smso:main:cor}). 

\begin{cor}[Thm.\ \ref{thm:smso:main}.(\ref{thm:smso:main:cor})]
\label{cor:real:main}
Consider a formula $\varphi(\vec X;\vec Y)$
with only $\vec X,\vec Y$ free, where
$\vec X = X_1,\dots,X_p$ and $\vec Y = Y_1,\dots,Y_q$.
Given a derivation $\Der D$ in $\SMSO$ of
$\thesis \ex {\vec Y} \varphi(\vec X;\vec Y)$,
we have $F_{\Der D} \iso \pair{\vec \seqbis,G}$
where $\vec\seqbis = \seqbis_1,\dots,\seqbis_q$
with $\seqbis_i : \two^p \longto_\Mealy \two$
and
$\Std \models \varphi(\vec\seq,\vec \seqbis(\vec\seq))$
for all $\vec\seq \in {(\two^\omega)}^p \iso {(\two^p)}^\omega$.
\end{cor}

The proof of Thm.\ \ref{thm:real:adeq} goes by induction on derivations.
Most of the rules of $\SMSO$ are straightforward,
except the synchronous comprehension rule, that we discuss first.
Adequacy for synchronous comprehension
follows from the existence of finite-state characteristic functions
for bounded formulae (Lem.\ \ref{lem:synch:char})
and from
the following lemmas,
which allow us, given a synchronous function
$\seqbis$ $y$-represented by $\hat\varphi$,
to lift a realizer of
$\psi[\hat\varphi[y]/Y]$
to a realizer of $\exists Y \psi$.

\begin{lem}[Substitution Lemma for Synchronous Comprehension]%
\label{lem:real:ca:subst}
Let $\vec x = x_1,\dots,x_\ell$ and $\vec X = X_1,\dots,X_p$.
Let $\hat\varphi$ be a
formula with free variables among $y,\vec X$, 
and which
$y$-represents $\seqbis : \two^p \longto_\Mealy\two$.
Then for every formula $\psi$ with
free variables among $\vec x,\vec X,Y$,
for all $\vec k \in \NN^\ell$
and all $\vec\seq \in {(\two^\omega)}^p \iso {(\two^p)}^\omega$
we have
\[
(\vec k,\vec\seq) \in \Lang(\I{\psi[\hat\varphi[y]/Y]}_{\vec x,\vec X})
\qquad\iff\qquad
(\vec k,\vec\seq, \seqbis(\vec\seq)) \in
\Lang(\I\psi_{\vec x,\vec X Y})
\]
\end{lem}

\begin{proof}
By induction on $\psi$.
\begin{itemize}
\item
If $\psi$ is an atomic formula not of the form $(x \In Y)$,
then $\psi[\hat\varphi[y]/Y] = \psi$
and the result is trivial.

\item
If $\psi$ is of the form $(x_i \In Y)$,
then
$\psi[\hat\varphi[y]/Y] = \hat\varphi[x_i/y]$.
Since $\hat\varphi$ $y$-represents $\seqbis$,
by~(\ref{eq:synch:repr}) (Def.\ \ref{def:synch:repr}),
for all $\vec k \in \NN^\ell$
and all $\vec\seq \in {(\two^\omega)}^p$
we have
\[
\seqbis(\vec\seq)(k_i) = 1
\qquad\text{iff}\qquad
\Std \models \hat\varphi[k_i/z,\vec\seq/\vec X]
\]
that is
\[
\Std \models k_i \In \seqbis(\vec\seq)
\qquad\text{iff}\qquad
\Std \models \hat\varphi[k_i/z,\vec\seq/\vec X]
\]

\noindent
Then we are done since it follows from Prop.\ \ref{prop:real:cor} that
\[
\Std \models \hat\varphi[k_i/z,\vec\seq/\vec X]
\qquad\text{iff}\qquad
(\vec k,\vec\seq) \in \Lang(\I{\hat\varphi[x_i/y]}_{\vec x,\vec X})
\]

\item
If $\psi = \psi_1 \land \psi_2$,
then by Prop.\ \ref{prop:real:cart}
we have
\[
\begin{array}{r !{\qquad} l !{\quad=\quad} l}
& \Lang \big( \I{(\psi_1\land\psi_2)[\hat\varphi[y]/Y]} \big)
& \Lang \big( \I{\psi_1[\hat\varphi[y]/Y]} \big)
  ~\cap~
  \Lang \big( \I{\psi_2[\hat\varphi[y]/Y]} \big)

\\
  \text{and}

& \Lang \big( \I{\psi_1\land\psi_2} \big)
& \Lang \big( \I{\psi_1} \big)
  ~\cap~
  \Lang \big( \I{\psi_2} \big)
\end{array}
\]
and we conclude by induction hypothesis.

\item
The cases of $\psi$ of the form $\ex X \varphi$ or $\ex x \varphi$
are similar, using Prop.~\ref{prop:real:realproj}
instead of Prop.~\ref{prop:real:cart}.

\item
The case of $\psi$ of the form $\lnot \varphi$ follows from
Prop.\ \ref{prop:real:compl} and the induction hypothesis.
\qedhere
\end{itemize}
\end{proof}

\begin{lem}[Lifting Lemma for Synchronous Comprehension]%
\label{lem:real:ca:lift}
Let $\vec x = x_1,\dots,x_\ell$ and $\vec X = X_1,\dots,X_p$.
Let $\hat\varphi$ be a
formula with free variables among $y,\vec X$,
and which
$y$-represents $\seqbis : \two^p \longto_\Mealy\two$.
Then for every formula $\psi$ with
free variables among $\vec x,\vec X,Y$,
there is a finite-state synchronous function
\[
H \quad:\quad
\Moves(\psi[\hat\varphi[y]/Y])
\quad\longto_\Mealy\quad
\Moves(\psi)
\]
such that
for all $\vec k \in \NN^\ell$,
all $\vec\seq \in {(\two^\omega)}^p$
and all $\run \in {\Moves(\psi[\hat\varphi[y]/Y])}^\omega$,
we have
\begin{equation}
\label{eq:real:ca:subst:iso}
\run \real \I{\psi[\hat\varphi[y]/Y]}_{\vec x,\vec X}(\vec k,\vec\seq)
\qquad\imp\qquad
H(R) \real
  \I\psi_{\vec x,\vec X,Y}(\vec k,\vec\seq, \seqbis(\vec\seq))
\end{equation}
\end{lem}

\begin{proof}
By induction on $\psi$.
\begin{itemize}
\item
If $\psi$ is an atomic formula not of the form $(x \In Y)$,
then $\psi[\hat\varphi[y]/Y] = \psi$.
So we take the identity for $H$
and the result trivially follows.

\item If $\psi$ is of the form $(x_i \In Y)$,
then
$\psi[\hat\varphi[y]/Y] = \hat\varphi[x_i/y]$.
Since $\psi$ is deterministic,
we can take for $H$
the unique map $\Moves(\hat\varphi[x_i/y]) \to_\Mealy \Moves(x_i \In Y) = \one$,
and the result follows from
Lem.\ \ref{lem:real:ca:subst}.

\item If $\psi$ is of the form $\varphi_1 \land \varphi_2$
(resp.\@ $\ex X \varphi$, $\ex x \varphi$)
then we conclude by induction hypothesis and Prop.\ \ref{prop:real:cart}
(resp.\@ Prop.\ \ref{prop:real:realproj}).

\item If $\psi = \lnot \varphi$, then we have
$\Moves(\psi) = \Moves(\psi[\hat\varphi[y]/Y]) = \one$,
and $H$ is the identity.
We then conclude by Lem.\ \ref{lem:real:ca:subst}.
\qedhere
\end{itemize}
\end{proof}

\noindent
Adequacy for synchronous comprehension follows easily.

\begin{lem}[Adequacy of Synchronous Comprehension]%
\label{lem:real:ca}
Let $\psi$ with free variables among $\vec x,\vec X,Y$
and let $\hat\varphi$ be a
formula with free variables among $y,\vec X$ and which is uniformly bounded by $y$.
Then there is a finite-state realizer
$\psi[\hat\varphi[y]/Y] \real_{\vec x,\vec X} F : \ex Y \psi$,
effectively computable from $\psi$ and $\varphi$.
\end{lem}

\begin{proof}
Let $\seqbis$ $y$-represented by $\hat\varphi$
be given by Lem.\ \ref{lem:synch:char},
and let $H$ satisfying~(\ref{eq:real:ca:subst:iso})
be given by Lem.\ \ref{lem:real:ca:lift}.
It then directly follows from Prop.\ \ref{prop:real:realproj}
and Lem.\ \ref{lem:real:ca:lift} that
$\psi[\hat\varphi[y]/Y]
\real_{\vec x,\vec X}
\pair{\seqbis \comp \lift\pi ,H\comp \lift{\pi'}}
: \exists Y \psi$,
where $\pi,\pi'$ are suitable projections.
\end{proof}


We can finally prove of Thm.\ \ref{thm:real:adeq}.

\begin{proof}[Proof of Thm.\ \ref{thm:real:adeq}]
The proof is by induction on derivations.
Note that if $\vec\varphi \thesis_\SMSO \varphi$, then
the universal closure of the implication $\land\vec\varphi \limp \varphi$
holds in the standard model $\Std$.
In particular,
for all rules whose conclusion is of the form
$\vec\varphi \thesis \delta$
with $\delta$ deterministic,
it follows from Prop.\ \ref{prop:real:cor}
and~(\ref{eq:real:moves}) that
the unique $\Mealy$-map with codomain
$\Moves(\delta) \iso \one$ (and with appropriate domain) is a realizer.
This handles the rules of \emph{negative} \emph{comprehension}
(\ref{eq:smso:negca}),
\emph{deterministic} \emph{induction} (\ref{eq:smso:ind})
and of elimination of double negation on deterministic formulae
(\ref{eq:smso:dnedet}).
This also handles all the rules of Fig.\ \ref{fig:mso:arith},
excepted the rules of elimination of equality
\[
\dfrac{\vec\varphi \thesis \varphi[y/x] \qquad \vec\varphi \thesis y \Eq z}
  {\vec\varphi \thesis \varphi[z/x]}
\]

\noindent
as well as the rules
$\overline{\vec\varphi \thesis \ex y \Zero(y)}$
and
$\overline{\vec\varphi \thesis \ex y \Succ(x,y)}$.
For the latter, we use the DMM depicted in
Fig.\ \ref{fig:prelim:mealy} (left)
(Ex.\ \ref{ex:prelim:mealy}.(\ref{ex:prelim:mealy:succ})) 
together with the fact that $\Succ(-,-)$ is deterministic.
The case of the former is similar and simpler.
As for elimination of equality, we take as realizer of the conclusion
the realizer of the left premise.
This realizer is trivially correct if there is no
realizer of the assumptions $\vec\varphi$.
Otherwise the result follows from Prop.\ \ref{prop:real:cor}
since the right premise ensures that the individual variables $y$ and $z$
are interpreted by the same natural number.
Adequacy for synchronous comprehension is given by Lem.\ \ref{lem:real:ca}.
It remains to deal with the rules of Fig.\ \ref{fig:mso:ded}.
The first two rules follow from the fact that each $\Aut_\Sigma$
is a category with finite products
(Prop.\ \ref{prop:real:cart}).
The rules for $\lnot/\False$ are trivial since their
conclusions are of the form $\vec\varphi \thesis \delta$.
The rules for conjunction follow from
Prop.\ \ref{prop:real:cart}
and those for existential quantifications follow from
Prop.\ \ref{prop:real:realproj}.
\end{proof}

\section{Indexed Structure on Automata}%
\label{sec:fib}

\noindent
In~\S\ref{sec:real} we have defined one category $\Aut_\Sigma$
for each alphabet $\Sigma$.
These categories are actually related by \emph{substitution functors}
arising from $\Mealy$-morphisms, inducing an \emph{indexed}
(or \emph{fibred}) structure.
Substitution functors are a basic notion of categorical logic,
which allows for categorical axiomatizations of quantifications.
We refer to \eg~\cite[Chap.\@ 1]{jacobs01book} for background.

We present here the fibred structure of the categories $\Aut_{(-)}$
and show that the existential quantifiers $\exists_{(-)}$ and Cartesian
product $(-)\times(-)$ of~\S\ref{sec:real:op}
satisfy the expected properties of existential quantifiers and conjunction
in categorical logic.
These properties essentially correspond to the adequacy of the logical
rules of Fig.\ \ref{fig:mso:ded} that do not mention negation ($\lnot$)
nor falsity ($\False$).
Although the fibred structure is not technically necessary to prove
the adequacy of our realizability model,
following such categorical axiomatization was a guideline in its design.
Besides, categorical logic turns out to be an essential tool
when dealing with generalizations to (say) alternating automata.

\subsection{The Basic Idea}%
\label{sec:fib:base}
Before entering the details, let us try to explain the main ideas
in the
usual setting of first-order logic
over a manysorted individual language.
The categorical semantics of existential quantifications
is given by an adjunction
\begin{equation}
\label{eq:fib:exadjsimple}
\text{
  \AXC{$\ex x \varphi(x) \thesis \psi$}
  \doubleLine
  \RightLabel{\quad ($x$ not free in $\psi$)}
  \UIC{$\varphi(x) \thesis \psi$}
\DisplayProof}
\end{equation}

\noindent
This adjunction
induces a bijection
between (the interpretations of) proofs
of the sequents $\varphi(x) \thesis \psi$
and $\ex x \varphi(x) \thesis \psi$,
that we informally denote
\[
\varphi(x) \thesis \psi
\qquad\iso\qquad
\ex x \varphi(x) \thesis \psi
\]

\noindent
Now, in general the variable $x$ occurs free in $\varphi$.
As a consequence,
in order to properly formulate~(\ref{eq:fib:exadjsimple})
one should be able to interpret sequents of the form
$\varphi(x) \thesis \psi$ with free variables.
More generally, the formulae $\varphi$ and $\psi$
should be allowed to contain free variables distinct from~$x$.

The idea underlying the
general method (see \eg~\cite{jacobs01book} for details),
is to first devise a base category $\cat B$ of individuals,
whose objects interpret products of sorts
of the individual language,
and whose maps from say $\iota_1 \times \cdots \times \iota_m$
to $o_1 \times \cdots \times o_n$
represent $n$-tuples $(t_1,\dots,t_n)$ of terms $t_i$
of sort $o_i$
whose free variables are among
$x_{\iota_1},\dots,x_{\iota_m}$,
with $x_{\iota_j}$ of sort $\iota_j$.
Then,
for each object $\iota = \iota_1 \times \cdots \times \iota_m$ of $\cat B$,
one devises a category $\cat E_\iota$ whose objects represent formulae
with free variables among $x_{\iota_1},\dots,x_{\iota_m}$,
and whose morphisms interpret proofs.
Furthermore, $\cat B$-morphisms 
\[
t = (t_1,\dots,t_n) \quad:\quad
\iota_1 \times \cdots \times \iota_m
\quad\longto\quad
o_1 \times \cdots \times o_n
\]
induce \emph{substitution functors}
\[
\subst t \quad:\quad \cat E_{o_1 \times \cdots \times o_n}
\quad\longto\quad
\cat E_{\iota_1 \times \cdots \times \iota_m}
\]

\noindent
The functor $\subst t$ takes
(the interpretation of) a formula $\varphi$
whose free variables are among
$y_{o_1},\dots,y_{o_n}$
to (the interpretation of) the formula
$\varphi[t_1/y_{o_1},\dots,t_n/y_{o_n}]$
with free variables among
$x_{\iota_1},\dots,x_{\iota_m}$.
Its action on the morphisms of $\cat E_{o_1 \times \cdots \times o_n}$
allows us to interpret the \emph{substitution rule} 
\[
\dfrac{\varphi \thesis \psi}
  {\varphi[t_1/y_{o_1},\dots,t_n/y_{o_n}] \thesis \psi[t_1/y_{o_1},\dots,t_n/y_{o_n}]}
\]

\noindent
In very good situations, the operation $\subst{(-)}$
is itself functorial.
Among the morphisms of $\cat B$,
one usually requires the existence of projections, say
\[
\pi
\quad:\quad
o \times \iota
\quad\longto\quad
o
\]
Projections induce substitution functors, called \emph{weakening} functors
\[
\subst\pi
\quad:\quad
\cat E_{o}
\quad\longto\quad
\cat E_{o \times \iota}
\]
which simply allow us to see formula $\psi(y_{o})$ with free variable $y_{o}$
as a formula $\psi(y_{o},x_\iota)$ with free variables among $y_{o},x_\iota$
(but with no actual occurrence of $x_\iota$).
Then the proper formulation of~(\ref{eq:fib:exadjsimple})
is that existential quantification over $x_\iota$ is a functor
\[
\ex {x_\iota} (-)
\quad:\quad
\cat E_{o \times \iota}
\quad\longto\quad
\cat E_{o}
\]
which is left-adjoint to $\subst\pi$:
\[
\text{
  \AXC{$\ex {x_\iota} \varphi(x_\iota,y_o)
    \thesis \psi(y_o)$}
  \doubleLine
  \UIC{$\varphi(x_\iota,y_o) \thesis \subst{\pi}(\psi)(x_\iota,y_o)$}
\DisplayProof}
\]
(where
$x_{\iota}$ does not occur free in $\psi$
since $\psi$ is assumed to be (interpreted as) an object of $\cat E_{o}$,
thus replacing the usual side condition).
Universal quantifications are dually axiomatized as right adjoints
to weakening functors.
In both cases, the adjunctions are subject to additional conditions
(called the \emph{Beck-Chevalley} conditions) which ensure that
they are preserved by substitution.

\subsection{Substitution}%
\label{sec:fib:subst}
So far, for each alphabet $\Sigma$ we have defined a category $\Aut_\Sigma$
of uniform automata over $\Sigma$.
Following~\S\ref{sec:fib:base},
different categories $\Aut_\Sigma$, $\Aut_\Gamma$ can be related
by means of $\Mealy$-morphisms $F : \Sigma \to \Gamma$.
This relies on a very simple \emph{substitution operation}
on automata, generalizing the substitution operation presented
in Prop.\ \ref{prop:real:subst}.

\begin{defi}[Substitution]%
\label{def:fib:subst}
Given a DMM $\At M : \Sigma \to \Gamma$ as in Def.\ \ref{def:mealy}
and an automaton $\At A : \Gamma$ as in Def.\ \ref{def:real:aut},
the automaton $\At A\lift{\At M} : \Sigma$
is defined as follows:
\[
\At A\lift{\At M}
\quad\deq\quad
(Q_{\At A} \times Q_{\At M},\, (\init q_{\At A},\init q_{\At M}),\,
\Moves(\At A),\,
\trans_{\At A\lift{\At M}},\, \Omega_{\At A\lift{\At M}})
\]
where
\[
\trans_{\At A\lift{\At M}}
\quad:\quad
Q_{\At A} \times Q_{\At M} \times \Sigma
\times
\Moves(\At A)
\quad\longto\quad
Q_{\At A} \times Q_{\At M}
\]
is defined as
\[
\trans_{\At A\lift{\At M}}((q_{\At A},q_{\At M}), \al a, m)
~~\deq~~
(\trans_{\At A}(q_{\At A}, \al b,m),\, q'_{\At M})
\qquad\text{with}\qquad
(q'_{\At M},\al b) ~~\deq~~ \trans_{\At M}(q_{\At M},\al a)
\]
and where
${(q_k,q'_k)}_{k \in \NN} \in \Omega_{\At A\lift{\At M}}$
iff ${(q_k)}_{k \in \NN} \in \Omega_{\At A}$.
\end{defi}

Note the reversed direction of the action of $\At M : \Sigma \to \Gamma$:
the substitution operation $(-)\lift{\At M}$ takes an automaton
over $\Gamma$ to an automaton over $\Sigma$.
Substitutions of the form $\At A\lift{\At M}$
can be seen as generalizations
of the substitutions presented in Prop.\ \ref{prop:real:subst}:
Given a function $\al f : \Sigma \to \Gamma$,
the automaton $\At A\lift{\al f}$ of Prop.\ \ref{prop:real:subst}
is isomorphic (in $\Aut_{\Sigma}$)
to the automaton $\At A\lift{\At M_{\al f}}$ obtained by applying
Def.\ \ref{def:fib:subst} to the one-state DMM inducing
the $\Mealy$-morphism $\lift{\al f} : \Sigma \to_{\Mealy} \Gamma$
of Rem.\ \ref{rem:prelim:lift}.

We now characterize the language of $\At A\lift{\At M}$.
To this end, it is useful to note that $\Sigma^\omega$
is in bijection with the set of synchronous functions
$\one^\omega \to \Sigma^\omega$.

\begin{prop}
Given a DMM $\At M : \Sigma \to \Gamma$ and
an automaton $\At A : \Gamma$,
for $\seq \in \Sigma^\omega$
we have:
\[
\seq \in \Lang(\At A\lift{\At M})
\qquad\text{iff}\qquad
F_{\At M}\comp \seq \in \Lang(\At A)
\]
where $F_{\At M} \comp \seq$ is the composition of the synchronous
function $F_{\At M}$ induced by $\At M$
with $\seq$ seen as a synchronous function $\one^\omega \to \Sigma^\omega$.
\end{prop}

Given $\At M : \Sigma \to \Gamma$,
an important property of the substitution operation
$(-)\lift{\At M}$ is that it induces a functor
$\Aut_{\Gamma} \to \Aut_\Sigma$.
The action of this functor on objects of $\Aut_\Gamma$
has just been defined.
Given a morphism $\At A \real F : \At B$ of $\Aut_{\Gamma}$,
the morphism $\At A\lift{\At M} \real F\lift{\At M} : \At B\lift{\At M}$
is the finite-state synchronous function
\[
F\lift{\At M}
\quad:\quad
\Sigma \times \Moves(\At A) \quad\longto_\Mealy\quad \Moves(\At B)
\]
taking $(\seq,\run)$ to $F(F_{\At M}(\seq),\run)$,
where $F_{\At M}$ is the finite-state synchronous function induced
by $\At M$.
It is easy to see that the action of $(-)\lift{\At M}$
on morphisms preserves identities and composition.

\subsection{Categorical Existential Quantifications}%
\label{sec:fib:quant}
Recall from~\S\ref{sec:real:op} that uniform automata are equipped with existential
quantifications, given by an adaption of
the usual projection operation on non-deterministic automata.
Given $\At A : \Sigma \times \Gamma$,
we defined
$\exists_\Gamma \At A : \Sigma$ as
\[
\exists_\Gamma \At A
~\deq~
(Q_{\At A} \,,\, \init q_{\At A} \,,\, \Gamma \times \Moves(\At A)
\,,\, \trans \,,\, \Omega_{\At A})
\quad\text{with}\quad
\trans(q,\al a,(\al b,u)) ~\deq~ \trans_{\At A}(q,(\al a,\al b),u)
\]

\noindent
We are now going to see that $\exists_{(-)}$
is an existential quantification in the usual categorical sense
of \emph{simple coproducts} (see \eg~\cite[Def.\@ 1.9.1]{jacobs01book}).
First, the \emph{weakening functors}
\[
(-)\lift\pi \quad:\quad
\Aut_{\Sigma}
\quad\longto\quad
\Aut_{\Sigma \times \Gamma}
\]
alluded to in~\S\ref{sec:fib:base}
are the substitution functors induced by projections
(see also Ex.\ \ref{ex:real:subst}):
\[
\lift\pi \quad:\quad \Sigma \times \Gamma \quad\longto_\Mealy\quad \Sigma
\]

\noindent
We can now state the first required property, namely that
$\exists_\Gamma$ induces a functor left adjoint to
$(-)\lift\pi$.

\begin{prop}%
\label{prop:fib:quant:adj}
Each existential quantifier $\exists_\Gamma$ induces
a functor $\Aut_{\Sigma \times \Gamma} \to \Aut_\Sigma$
which is left-adjoint to the weakening functor
$(-)\lift{\pi} : \Aut_\Sigma \to \Aut_{\Sigma \times \Gamma}$.
\end{prop}

\begin{proof}
Fix alphabets $\Sigma$ and $\Gamma$.
According to~\cite[Thm.\@ IV.1.2.(ii)]{maclane98book},
we have to show that for each automaton $\At A : \Sigma \times \Gamma$,
there is an $\Aut_{\Sigma \times \Gamma}$-morphism
\[
\eta_{\At A} \quad:\quad
\At A
\quad\longto\quad
(\exists_\Gamma \At A)\lift\pi
\]

\noindent
satisfying the following universal property:
for each automaton $\At B : \Sigma$
and each
$\Aut_{\Sigma\times \Gamma}$-morphism
\[
F \quad:\quad
\At A
\quad\longto\quad
\At B\lift\pi
\]

\noindent
there is a unique $\Aut_\Sigma$-morphism
\[
H \quad:\quad
\exists_\Gamma \At A
\quad\longto\quad
\At B
\]

\noindent
such that we have
\[
\xymatrix@C=40pt@R=40pt{
  \At A
  \ar[r]^(0.4){\eta_{\At A}} 
  \ar[d]_{F} 
& (\exists_\Gamma \At A)\lift\pi
  \ar[dl]^{H\lift\pi} 
\\
  \At B\lift\pi
}
\]

Note that $\eta_{\At A}$ must be an $\Mealy$-morphism
\[
\eta_{\At A}
\quad:\quad
(\Sigma \times \Gamma) \times \Moves(\At A)
\quad\longto_\Mealy\quad
\Gamma \times \Moves(\At A)
\]

\noindent
We let $\eta_{\At A}$ be the $\Mealy$-morphism
induced by the usual projection
$\Sigma \times \Gamma \times \Moves(\At A)
\to
\Gamma \times \Moves(\At A)$.
Given $\At A \real F : \At B\lift\pi$,
we are left with the following trivial fact:
there is a unique
$\exists_\Gamma \At A \real H : \At B$
such that
\[
\forall \seq \in \Sigma^\omega,~
\forall \seqbis \in \Gamma^\omega,~
\forall \run \in {\Moves(\At A)}^\omega,\quad
F(\pair{\seq,\seqbis},\run)
\quad=\quad
H(\seq,\pair{\seqbis,\run})
\qedhere
\]
\end{proof}

The Beck-Chevalley condition of~\cite[Def.\@ 1.9.1]{jacobs01book}
asks for the following isomorphism in $\Aut_{\Delta}$,
where $\At A : \Sigma \times \Gamma$
and $F : \Delta \to_\Mealy \Sigma$:
\[
(\exists_\Gamma \At A)\lift{\At M_F}
\quad\iso\quad
\exists_\Gamma(\At A\lift{\At M_{F \times \id_\Gamma}})
\]
This isomorphism follows from the fact that the
two above automata have the same set of moves
(namely $\Gamma \times \Moves(\At A)$).

\subsection{Categorical Conjunction}%
\label{sec:fib:conj}
Recall from~\S\ref{sec:real:op} that each category $\Aut_\Sigma$ has
Cartesian products, which interpret conjunction, a necessary feature to
interpret a sequent as a morphism from the conjunct of its premises to its
conclusion. In the setting of categorical logic, it remains to be shown that
these products are \emph{fibred} in the sense of~\cite[Def.\@ 1.8.1]{jacobs01book},
i.e.\@ that they are preserved by substitution.

\begin{prop}%
\label{lem:prop:prod:fibred}
Given automata $\At A, \At B : \Gamma$ and a DMM
$\At M : \Sigma \to \Gamma$, the product $\At A\lift{\At M} \times \At B\lift{\At M}$
is isomorphic to $(\At A \times \At B)\lift{\At M}$
in $\Aut_\Sigma$.
\end{prop}

\begin{proof}
The isomorphism trivially follows from the fact that
\[
\Moves(\At A\lift{\At M} \times \At B\lift{\At M})
\quad=\quad
\Moves(\At A) \times \Moves(\At B)
\quad=\quad
\Moves((\At A \times \At B)\lift{\At M})
\qedhere
\]
\end{proof}


\subsection{Indexed Structure}%
\label{sec:fib:indexed}
Thanks to the substitution operation discussed in~\S\ref{sec:fib:subst},
each $\Mealy$-morphism $F : \Sigma \to \Gamma$ induces a functor
$(-)\lift{\At M_F} : \Aut_\Gamma \to \Aut_\Sigma$,
where $\At M_F$ is a \emph{chosen} DMM inducing $F$.
As usual in categorical logic, we would like to extend substitution to
a functor $\subst{(-)} : \Mealy^\op \to \Cat$
taking alphabets $\Sigma$ to categories $\Aut_\Sigma$,
and $\Mealy$-morphisms $F : \Sigma \to \Gamma$ to functors
$\Aut_\Gamma \to \Aut_\Sigma$.
In order for $\subst{(-)}$ to be a functor, it should preserve identities and
composition.
In particular, given an automaton $\At A: \Sigma$,
for all
$\Mealy$-maps
$G : \Delta \to \Gamma$ and $F : \Gamma \to \Sigma$
we should have
\begin{equation}
\label{eq:fib:pseudo}
\At A ~~=~~ \At A\lift{\At M_{\Id_\Sigma}}
\qquad\text{and}\qquad
(\At A\lift{\At M_{F}})\lift{\At M_{G}}
~~=~~
\At A\lift{\At M_{F \comp G}}
\end{equation}

\noindent
But we see no reason for this to be possible.
In particular there is no reason for the DMM $\At M_{F \comp G}$
chosen to induce $F \comp G$ to be a product of $\At M_F$ and $\At M_G$.
However, since $\At A\lift{\At M}$ always has the same moves as $\At A$,
we actually get~(\ref{eq:fib:pseudo}) modulo isomorphisms.

This is a usual situation in categorical logic.
It is indeed customary to relax the requirement
of $\subst{(-)}$ to be a functor, and only ask it to be a \emph{pseudo} functor,
\ie\@ a functor for which identities and composition are only preserved up to
natural isomorphisms, subject to some specific coherence conditions
(see \eg~\cite[Def.\@ 1.4.4]{jacobs01book}).
The required natural isomorphisms have the form
\begin{equation}
\label{eq:fib:nat}
\begin{array}{r !{\quad\stackrel\iso\longto\quad} l}
  \eta_\Sigma
  \quad:\quad
  \Id_{\Aut_\Sigma}
& (-)\lift{\At M_{\Id_\Sigma}}
\\
  \mu_{G,F}
  \quad:\quad
  (-)\lift{\At M_F}\lift{\At M_G}
& (-)\lift{\At M_{F \comp G}}
\end{array}
\end{equation}

\noindent
Since $\At A$ and $\At A\lift{\At M}$ have the same moves,
we can take for each components of $\eta_\Sigma$ and $\mu_{F,G}$
synchronous functions acting as identities on runs.
It then follows that all the required diagrams commute.

We now proceed to the formal construction.
Fix for each $\Mealy$-morphism $F : \Sigma \to_\Mealy \Gamma$
a chosen DMM $\At M_F$ inducing $F$.
For each $\At A : \Sigma$, and each
$\Mealy$-morphisms
$G : \Delta \to_\Mealy \Gamma$ and $F : \Gamma \to_\Mealy \Sigma$,
we let
\[
\At A \real \eta_{\Sigma,\At A} : \At A\lift{\At M_{\Id_\Sigma}}
\qquad\text{and}\qquad
\At A \lift{\At M_F}\lift{\At M_G} \real \mu_{G,F,\At A} : \At A\lift{\At M_{F\comp G}}
\]

\noindent
be given by
\[
\begin{array}{r c c c c}
  \eta_{\Sigma,\At A}
& :
& \Sigma \times \Moves(\At A)
& \longto
& \Moves(\At A)
\\
&
& (\seq,\run)
& \longmapsto
& \run
\end{array}
\qquad\text{and}\qquad
\begin{array}{r c c c c}
  \mu_{G,F,\At A}
& :
& \Sigma \times \Moves(\At A)
& \longto
& \Moves(\At A)
\\
&
& (\seq,\run)
& \longmapsto
& \run
\end{array}
\]

\noindent
The following 
says that
the coherence conditions required for structure maps
of pseudo-functors (see \eg~\cite[Def.\@ 1.4.4]{jacobs01book})
are 
met by $\eta_{\Sigma,\At A}$ and $\mu_{F,G,\At A}$.
The proof is trivial.

\begin{prop}
The morphisms $\eta_\Sigma$ and $\mu_{F,G}$ defined above are
natural isomorphisms as in~(\ref{eq:fib:nat}).
Moreover, for each automaton $\At A:\Sigma$
and each $\Mealy$-maps $F,G,H$ of appropriate domains and codomains,
the two diagrams of Fig.\ \ref{fig:fib:indexed} commute.
\end{prop}

The assignment
$\subst{(-)} : \Mealy^\op \to \Cat$
taking the alphabet $\Sigma$ to the category $\Aut_\Sigma$ and
the morphism $F : \Gamma \to_\Mealy \Sigma$
to the functor $(-)\lift{\At M_F} : \Aut_\Sigma \to \Aut_\Gamma$
is thus a pseudo-functor.

\begin{figure}
\[
\xymatrix@C=70pt{
& \At A\lift{\At M_F}
  \ar@{=}[d]
  \ar[dl]_{\eta_{\Gamma,\At A\lift{\At M_F}}} 
  \ar[dr]^{\eta_{\Sigma,\At A}\lift{\At M_F}} 
\\
  \At A\lift{\At M_F}\lift{\At M_{\Id_\Gamma}}
  \ar[r]_{\mu_{\Id_\Gamma,F,\At A}} 
& \At A\lift{\At M_F}
& \At A\lift{\At M_{\Id_\Sigma}}\lift{\At M_F}
  \ar[l]^{\mu_{F,\Id_\Sigma,\At A}} 
}
\]
\[
\xymatrix@C=70pt{
  \At A\lift{\At M_F}\lift{\At M_G}\lift{\At M_H}
  \ar[r]^{\mu_{G,F,\At A}\lift{\At M_H}} 
  \ar[d]_{\mu_{H,G,\At A\lift{\At M_F}}} 
& \At A\lift{\At M_{F\comp G}}\lift{\At M_H}
  \ar[d]^{\mu_{H,F \comp G,\At A}} 
\\
  \At A\lift{\At M_F}\lift{\At M_{G\comp H}}
  \ar[r]_{\mu_{G\comp H,F,\At A}} 
& \At A\lift{F \comp G \comp H}
}
\]
\caption{Coherence Diagrams for the Structure Maps of
$\subst{(-)} : \Mealy^\op \to \Cat$.\label{fig:fib:indexed}}
\end{figure}

\section{Conclusion}

\noindent
In this paper, we
revisited Church's synthesis via
an automata-based realizability interpretation of
an intuitionistic proof system $\SMSO$ for $\MSO$ on $\omega$-words,
and we
demonstrated that our approach is sound and complete,
in the sense of Thm.\ \ref{thm:smso:main}.
As it stands, this approach must still pay the price of the non-elementary
lower-bound for the translation of $\MSO$ formulae over finite words
to DFAs
(see Rem.\ \ref{rem:synch:nelb})
and the system $\SMSO$
is limited by its set of connectives and its restricted induction scheme.

\subsection*{Further Works}
First, the indexed structure (\S\ref{sec:fib:indexed})
induced by the substitution operation of~\S\ref{sec:fib:subst}
suggests that in our context, it may be profitable
to work in a conservative extension of $\MSMSO$,
with one function symbol for each Mealy machine
together with defining axioms of the form~(\ref{eq:synch:form:repr}).
In particular, this could help mitigate Rem.\ \ref{rem:synch:nelb}
by giving the possibility, in the synchronous comprehension scheme of $\SMSO$,
to give a term for a Mealy machine rather than the $\MSO$-formula
representing it.
We expect this to give better lower bounds \wrt\@ completeness 
(for each solvable instance of Church's synthesis, to provide proofs
with realizers of a reasonable complexity).

Second,
following the approach of~\cite{riba16dialaut}, $\SMSO$ could be extended
with primitive universal quantifications and implications
as soon as one goes
to a \emph{linear} deduction system.
Among outcomes of going 
to a linear deduction system,
following~\cite{riba16dialaut}
we expect similar proof-theoretical properties as with the usual
\emph{Dialectica} interpretation (see \eg~\cite{kohlenbach08book}),
such as realizers of linear Markov rules
and choices schemes.
Also, having primitive universal quantifications may allow us
to take benefit of the reductions of $\MSO$ to its negative fragment,
as provided by the \emph{Safraless} approaches to
synthesis~\cite{kv05focs,kpv06cav,fjr11fmsd}.

Obtaining a good handle of induction in $\SMSO$ is more complex.
One possibility to have finite-state realizers for a more general induction rule
would be to rely on saturation techniques for regular languages.
Another possibility, which may be of practical interest, is to follow
the usual Curry-Howard approach and allow for possibly infinite-state realizers.

Another direction of future work is to incorporate specific reasoning principles
on Mealy machines.
For instance, a possibility could be
to base our deduction system on a complete equational theory
for Mealy machines.

\section*{Acknowledgment}
The authors would like to thank to anonymous referees
for their thorough readings of previous versions of this paper,
which helped a lot in raising its quality.

\bibliographystyle{alpha}
\bibliography{bibliographie}

\appendix

\begin{figure}[h!]
\begin{description}
\item[$\Succ$ is the Successor for $\Lt$]
\[
\forall x,y
\big[
\Succ(x,y)
\quad\longliff\quad
\big(x \Lt y ~\land~
\lnot \exists z(x \Lt z \Lt y) \big)
\big]
\]

\item[Strict Linear Order Axioms]
\[
\lnot(x \Lt x)
\qquad\qquad
(x \Lt y \Lt z \limp x \Lt z)
\qquad\qquad
(x \Lt y ~\lor~ x \Eq y ~\lor~ y \Lt x)
\]

\item[Predecessor and Unboundedness Axioms]
\[
\forall x
\big[
\exists y(y \Lt x)
~~\longlimp~~
\exists y \Succ(y,x)
\big]
\qquad\qquad
\forall x\exists y(x \Lt y)
\]
\end{description}
\caption{The Arithmetic Axioms of~\cite{riba12ifip}.\label{fig:app:prelim:axarith}}
\end{figure}

\section{Completeness of \texorpdfstring{$\MSO$}{MSO} (Thm.\ \ref{thm:mso:compl})}%
\label{sec:app:prelim}

In this Appendix, we provide the missing details
to deduce the completeness of our axiomatization of $\MSO$
(Thm.\ \ref{thm:mso:compl}) from~\cite{riba12ifip}.
The arithmetic axioms of~\cite{riba12ifip}
expressed with $\Lt$ and $\Succ$
are presented in Fig.\ \ref{fig:app:prelim:axarith},
where
\[
(x \Lt y)
\quad\deq\quad
\left(x \Leq y ~\land~ \lnot(x \Eq y)\right)
\]

\noindent
The axioms of Fig.\ \ref{fig:app:prelim:axarith} follow from
Lem.\ \ref{lem:prelim:mso:lem} (Fig.\ \ref{fig:prelim:mso:lem})
that we prove now.

\begin{proof}[Proof of Lem.\ \ref{lem:prelim:mso:lem}]
Let us recall the non-logical rules of $\MSO$ (omitting equality):
\begin{itemize}
\item $\Leq$ is a partial order:
\[
\dfrac{}{\vec\varphi \thesis x \Leq x}
\qquad
\dfrac{\vec\varphi \thesis x \Leq y \qquad \vec\varphi \thesis y \Leq z}
  {\vec\varphi \thesis x \Leq z}
\qquad
\dfrac{\vec\varphi \thesis x \Leq y \qquad \vec\varphi \thesis y \Leq x}
  {\vec\varphi \thesis x \Eq y}
\]

\item Basic $\Zero$ and $\Succ$ rules (total injective relations):
\[
\begin{array}{c}
\dfrac{}{\vec\varphi \thesis \ex y \Zero(y)}

\qquad\qquad

\dfrac{\vec\varphi  \thesis \Zero(x)
  \qquad \vec\varphi \thesis \Zero(y)}
  {\vec\varphi \thesis x \Eq y}

\\\\


\dfrac{}{\vec\varphi \thesis \ex y \Succ(x,y)}

\qquad\qquad

\dfrac{\vec\varphi \thesis \Succ(y,x)
  \qquad
  \vec\varphi \thesis \Succ(z, x)}
  {\vec\varphi \thesis y \Eq z}

\qquad\quad

\dfrac{\vec\varphi \thesis \Succ(x,y)
  \qquad \vec\varphi \thesis \Succ(x,z)}
  {\vec\varphi \thesis y \Eq z}
\end{array}
\]


\item
Arithmetic rules:
\[
\dfrac{\vec\varphi \thesis \Succ(x,y)
  \quad~~
  \vec\varphi \thesis \Zero(y)}
  {\vec\varphi \thesis \False}
\qquad~~
\dfrac{\vec\varphi \thesis \Succ(x, y)}{\vec\varphi \thesis x \Leq y}
\qquad~~
\dfrac{\vec\varphi \thesis \Succ(y,y')
  \quad~~
  \vec\varphi\thesis x \Leq y'
  \quad~~
  \vec\varphi \thesis\lnot(x \Eq y')}
  {\vec\varphi\thesis x \Leq y}
\]

\end{itemize}

\noindent
We now proceed to the proof of the properties listed in
Fig.\ \ref{fig:prelim:mso:lem}.
\begin{enumerate}
\item\label{eq:app:mso:ltirrefl}
$\thesis \lnot(x \Lt x)$

\begin{proof}
From reflexivity of equality.
\end{proof}

\item\label{eq:app:mso:lttrans}
$x \Lt y,y \Lt z \thesis x \Lt z$

\begin{proof}
We have
$x \Lt y,y \Lt z \thesis x \Leq z$
and
$x \Lt y,y \Lt z,x \Eq z \thesis \False$
by the partial order rules for $\Leq$.
\end{proof}

\item\label{eq:app:mso:succ}
$\Succ(x,y),x \Eq y \thesis \False$
\begin{proof}
By induction on $y$, we show
\[
\phi(y) \quad\deq\quad
\forall x(\Succ(x,y) \limp \lnot (x\Eq y))
\]
We have $\Zero(y) \thesis \phi(y)$ by the first arithmetic rule.
We now show $\phi(y),\Succ(y,y') \thesis \phi(y')$, that is
\[
\phi(y),\Succ(y,y'),\Succ(x,y'),x \Eq y' \thesis \False
\]
Note that
\[
\Succ(y,y'),\Succ(x,y'),x \Eq y' \thesis
x \Eq y ~\land~ y \Eq y' ~\land~ \Succ(x,y)
\]
From which follows that
\[
\phi(y), \Succ(y,y'),\Succ(x,y'),x \Eq y' \thesis \False
\qedhere
\]
\end{proof}

\item\label{eq:app:mso:unboundedness}
$\thesis \forall x\exists y (x\Lt y)$.

\begin{proof}
The basic and arithmetic rules above
ensure that every $x$ has a successor $y$, and that the successor satisfies $x \Le y$.
Thus, in combination with~(\ref{eq:app:mso:succ}), we get that $x \Lt y$.
\end{proof}

\item $\Succ(y,y'),x \Leq y, x \Eq y' \thesis \False$
\begin{proof}
We have
\[
\Succ(y,y'),x \Leq y, x \Eq y' \thesis y' \Leq y
\]
and by the partial order rules for $\Leq$
together with the second arithmetic rule, we have
\[
\Succ(y,y'),x \Leq y, x \Eq y' \thesis y' \Eq y
\]
and we conclude by~(\ref{eq:app:mso:succ}).
\end{proof}

\item\label{eq:app:mso:zeroleq}
$\Zero(x) \thesis x \Leq y$
\begin{proof}
By induction on $y$.
\end{proof}

\item\label{eq:app:mso:zeroleqzero}
$x \Leq y,\Zero(y) \thesis \Zero(x)$
\begin{proof}
By~(\ref{eq:app:mso:zeroleq}),
we have $\Zero(y) \thesis y \Leq x$
and we conclude by the partial order rule for $\Leq$.
\end{proof}

\item\label{eq:app:mso:zeromin}
$\forall y(x \Leq y) \thesis \Zero(x)$

\begin{proof}
We have
\[
\forall y(x \Leq y),\Zero(z) \thesis x \Leq z
\]

\noindent
Hence by~(\ref{eq:app:mso:zeroleqzero})
we get
\[
\forall y(x \Leq y),\Zero(z) \thesis \Zero(x)
\]
and we conclude by the basic rules for $\Zero$.
\end{proof}

\item\label{eq:app:mso:ltsucleq}
$x \Lt y,\Succ(x,x') \thesis x' \Leq y$
\begin{proof}
By induction on $y$, we show
\[
\phi(y) \quad\deq\quad \forall x,x'
\left(
x \Lt y \limp \Succ(x,x') \limp x' \Leq y
\right)
\]
First, the base case $\Zero(y) \thesis \phi(y)$
follows from the fact that
$\Zero(y), x \Lt y \thesis \False$
by~(\ref{eq:app:mso:zeroleqzero}).
For the induction step, we show
\[
\Succ(y,y'),\phi(y),x \Lt y',\Succ(x,x') \thesis x' \Leq y'
\]
We use the excluded middle on $x \Eq y$, and we are left with showing
\[
\Succ(y,y'),\phi(y),x \Lt y',\lnot(x \Eq y),\Succ(x,x') \thesis x' \Leq y'
\]
But by the arithmetic rules,
$\Succ(y,y'),x \Lt y' \thesis x \Leq y$,
so that
\[
\Succ(y,y'),\phi(y),x \Lt y',\lnot(x \Eq y),\Succ(x,x') \thesis x \Lt y
\]
But
\[
\phi(y),x \Lt y,\Succ(x,x') \thesis x' \Leq y
\]
and we are done.
\end{proof}

\item\label{eq:app:mso:sucmon}
$x \Leq y,\Succ(x,x'),\Succ(y,y') \thesis x' \Leq y'$
\begin{proof}
By induction on $z$ we show
\[
\phi(z) \quad\deq\quad
\forall x,x',y\left(
x \Leq y \limp \Succ(x,x') \limp \Succ(y,z) \limp x' \Leq z
\right)
\]
We trivially have
$\Zero(z)\thesis \phi(z)$.
We now show $\Succ(z,z'),\phi(z) \thesis \phi(z')$,
that is
\[
\Succ(z,z'),\phi(z), x \Leq y, \Succ(x,x'),\Succ(y,z') \thesis x' \Leq z'
\]
By the basic $\Zero$ and $\Succ$ rules, this amounts to show
\[
\phi(z), x \Leq z, \Succ(x,x'),\Succ(z,z') \thesis x' \Leq z'
\]
Now, using the excluded middle on $x \Eq z$, we are left with showing
\[
\phi(z), \Succ(x,x'),\Succ(z,z'),x \Lt z \thesis x' \Leq z'
\]
But by~(\ref{eq:app:mso:ltsucleq})
we have
\[
x \Lt z,\Succ(x,x') \thesis x' \Leq z
\]
and we are done.
\end{proof}

\item\label{eq:app:mso:ltleq}
$\thesis
\forall x \forall y
\left[
y \Lt x \quad\longliff\quad \exists z(y \Leq z ~\land~ \Succ(z,x))
\right]$
\begin{proof}
The right-to-left direction follows from~(\ref{eq:app:mso:succ}).
For the left-to-right direction, by induction on $x$,
we show
\[
\phi(x) \quad\deq\quad
\forall y\left(
y \Lt x ~\longlimp~ \exists z(y \Leq z ~\land~ \Succ(z,x))
\right)
\]

\noindent
For the base case
$\Zero(x)\thesis \phi(x)$,
by~(\ref{eq:app:mso:zeroleqzero})
we have
$\Zero(x), y \Leq x \thesis \Zero(y)$
and we conclude by the basic rules for $\Zero$.
For the induction step, we have to show
\[
\Succ(x,x'),\phi(x),y \Lt x' \thesis \exists z(y \Leq z ~\land~ \Succ(z,x'))
\]
By the last arithmetic rule,
\[
\Succ(x,x'),y \Lt x' \thesis y \Leq x
\]
and we are done.
\end{proof}

\item\label{eq:app:mso:ltlin}
$\thesis x \Lt y \lor x \Eq y \lor y \Lt x$

\begin{proof}
By induction on $x$, we show
\[
\phi(x) \quad\deq\quad
\forall y
\left(
x \Lt y \lor x \Eq y \lor y \Lt x
\right)
\]
The base case $\Zero[x] \thesis \phi(x)$
follows from~(\ref{eq:app:mso:zeroleq}).
For the induction step, we have to show
\[
\Succ(x,x'),\phi(x) \thesis
\forall y(x' \Lt y \lor x' \Eq y \lor y \Lt x')
\]
By induction on $y$ we show
$\Succ(x,x'),\phi(x) \thesis \psi[y,x']$
where
\[
\psi[y,x']
\quad\deq\quad
x' \Lt y \lor x' \Eq y \lor y \Lt x'
\]
The base case follows again from~(\ref{eq:app:mso:zeroleq}).
For the induction step, we have to show
\[
\Succ(x,x'),\Succ(y,y'),\phi(x),\psi[y,x'] \thesis
\forall y(x' \Lt y' \lor x' \Eq y' \lor y' \Lt x')
\]
and we are done since~(\ref{eq:app:mso:sucmon})
gives
\[
\forall x,x',y,y'\left(
x \Lt y \limp \Succ(x,x') \limp \Succ(y,y') \limp x' \Lt y'
\right)
\qedhere
\]
\end{proof}

\item\label{eq:app:mso:sucax}
$\thesis
\forall x,y
\left[
\Succ(x,y)
\quad\longliff\quad
\left(x \Lt y ~\land~
\lnot \exists z(x \Lt z \Lt y)\right)
\right]$

\begin{proof}
For the left-to-right direction, thanks to~(\ref{eq:app:mso:succ})
we are left with showing
\[
\Succ(x,y),x \Lt z,z \Lt y \thesis \False
\]

\noindent
But the last arithmetic rule gives $z \Leq x$ from $\Succ(x,y)$ and $z \Lt y$,
which together with $x \Lt z$ gives $z \Eq x$ by antisymmetry of $\Leq$,
contradicting $x \Lt z$.

Conversely, assume that $x \Lt y$ without any intermediate $z$.
By the basic rules for $\Succ$, we have $\Succ(x,z)$ for some $z$.
Note that $x \Lt z$ by~(\ref{eq:app:mso:succ}).
Since $x \Lt y$ it follows from~(\ref{eq:app:mso:ltsucleq})
that $z \Leq y$.
But this implies $z \Eq y$ as $\lnot(z \Lt y)$.
\end{proof}
\end{enumerate}

\noindent
This concludes the proof of Lem.\ \ref{lem:prelim:mso:lem}.
\end{proof}

\noindent
The linear order axioms
((\ref{eq:app:mso:ltirrefl}),
(\ref{eq:app:mso:lttrans}),
(\ref{eq:app:mso:ltlin})),
the successor axiom~(\ref{eq:app:mso:sucax}),
the unboundedness~(\ref{eq:app:mso:unboundedness})
and predecessor~(\ref{eq:app:mso:ltleq}) axioms are thus proved in our axiomatic.

Finally, we have to prove Lem.\ \ref{lem:prelim:mso:ind},
namely that strong induction is derivable in $\MSO$.
The proof holds no surprise.

\begin{proof}[Proof of Lem.\ \ref{lem:prelim:mso:ind}]
We have to show
\[
\forall x \big( \forall y(y \Lt x ~\limp~ X y) ~\longlimp~ X x \big)
~\thesis~ \all x X x
\]
By induction on $x$ we show
\[
\forall x \big( \forall y(y \Lt x ~\limp~ X y) ~\longlimp~ X x \big)
~\thesis~ \phi(x)
\]
where
\[
\phi(x)
~~\deq~~
\forall y(y \Leq x ~\limp~ X y)
\]

\noindent
The base case
\[
\forall x \big( \forall y(y \Lt x ~\limp~ X y) ~\longlimp~ X(x) \big)
,\, \Zero(x) ~\thesis~ \phi(x)
\]
is trivial since by~(\ref{eq:app:mso:zeroleq}),
\[
\Zero(x) \thesis \lnot \exists y(y \Lt x)
\]

\noindent
For the induction step, we have to show
\[
\forall x\big( \forall y(y \Lt x ~\limp~ X y) ~\longlimp~ X x \big)
,\,
\Succ(x,x')
,\,
\phi(x)
,\,
z \Leq x'
~\thesis~ X z
\]

\noindent
Notice that $\phi(x)$ is equivalent to
$\phi'(x') \deq \forall y(y \Lt x' \limp X y)$ thanks to~(\ref{eq:app:mso:ltleq}).
By~(\ref{eq:app:mso:ltlin}),
we have three subcase according to: 
\[
z \Lt x'
~\lor~
z \Eq x'
~\lor~
x' \Lt z
\]

\noindent
The first case enables us to use $\phi'(x')$ directly,
and the second one follows from the assumption
$\forall x(\forall y(y \Lt x \limp X y) \limp X x)$ together with $\phi'(x')$.
The last one leads to a contradiction using the antisymmetry of $\Leq$.
\end{proof}

\section{Internally Bounded Formulae (\S\ref{sec:synch:synchbounded})}%
\label{sec:app:synch:synchbounded}

We prove here the Splitting Lemma~\ref{lem:synch:splitting},
used in the proof of Thm.\ \ref{thm:synch:synchbounded}.
We consider formulae over the vocabulary
of~\cite{riba12ifip}, that is formulae given by the grammar
\[
\varphi,\psi \in \Lambda \quad\bnf\quad
\True \gs \False \gs
x \In X \gs x \Lt y
 \gs \lnot \varphi \gs \varphi \lor \psi \gs \ex X \varphi \gs \ex x \varphi
\]
Following~\S\ref{sec:prelim:ax}
(see also~\S\ref{sec:app:prelim}),
defining
the atomic formulae $\Eq$, $\Succ(-,-)$, $\Leq$ and $\Zero(-)$ as
\[
\begin{array}{r !{\quad\deq\quad} l}
  x \Eq y
& \all X(x \In X ~\limp~ y \In X)
\\
\Succ(x,y)
&
\left(x \Lt y ~\land~
\lnot \exists z(x \Lt z \Lt y)\right)
\\
x \Leq y
& (x \Lt y) \lor (x \Eq y)
\\
\Zero(x)
&
\all y(x \Leq y)
\end{array}
\]
we obtain for each formula in the sense of Fig.\ \ref{fig:mso:form}
an equivalent formula in $\Lambda$
(\wrt\@ the standard model $\Std$). 
%
%
Note that the Transfer
Lemma~\ref{lem:synch:transf} gives in particular that if
$\seq_0,\seq_1 \in \two^\omega \iso \Po(\NN)$ are disjoint, then
\begin{equation}
\label{eq:app:synch:sep}
\Std\models\quad
\exists X(\varphi_0\restr \seq_0 ~\land~ \varphi_1\restr \seq_1)
~~\longliff~~
\big( \exists X(\varphi_0\restr \seq_0) ~\land~ \exists X(\varphi_1 \restr \seq_1) \big)
\end{equation}

\begin{lem}[Splitting (Lem.\ \ref{lem:synch:splitting})]%
\label{lem:app:synch:splitting}
Let $\psi$ be a formula in $\Lambda$ and let
$z$ be an individual variable.
For every set of individual variables $V$ with $z \in V$,
one can produce a natural number $N$ and
two matching sequences of length $N$ of left formulae
${(L_j)}_{j < N}$
and right formulae ${(R_j)}_{j < N}$
such that the following holds:
\begin{itemize}
\item For every $j < N$,
$\FV^\indiv(L_j) \sle \FV^\indiv(\psi) \cap V$ and
$\FV^\indiv(R_j) \sle \FV^\indiv(\psi) \setminus V$.

\item
If $\FV^\indiv(\psi) = \{\vec x,z,\vec y\}$
with $V\cap\FV^\indiv(\psi) = \{\vec x,z\}$,
then
for all $n \in \NN$,
all $\vec a \leq n$ and all $\vec b > n$, we have
\[
\Std\models\quad
\psi[\vec a/\vec x,n/z,\vec b/\vec y]
~~\longliff~~
\bigdisj_{j < N}
L_j[\vec a/\vec x,n/z]\restr{[ - \Leq n ]}
~\land~
R_j[\vec b/\vec y]\restr{[ - \Gt n]}
\]
\end{itemize}
\end{lem}

\begin{proof}
The proof proceeds by induction on $\psi$.
The cases of $\True$ and $\False$ are trivial and omitted.
\begin{description}
\item[Case of $(x \Lt y)$]
We take $N \deq 1$ and
we define suitable left and right formulae according to $V$.
In each case the choice of $z$ is irrelevant.
\begin{itemize}
\item If $x, y \in V$, then $L_0 \deq \psi$ and $R_0 \deq \True$.
\item If $x,y \notin V$, then $L_0 \deq \True$ and $R_0 \deq \psi$.
\item If $x \in V$ and $y \notin V$, then $L_0 \deq R_0 \deq \True$.
\item If $y \in V$ and $x \notin V$, then $L_0 \deq R_0 \deq \False$.
\end{itemize}

\item[Case of $(x \In X)$]
We take $N \deq 1$.
Then one of the produced formula is $\psi$ and the other one
is $\True$ according to whether $x \in V$ or not.

\item[Case of $(\psi \lor \psi')$]
Let ${(L_j,R_j)}_{j<N}$ and ${(L'_k,R'_k)}_{k<N'}$ be obtained
by applying the induction hypothesis on $\psi$ and $\psi'$ respectivelly.
Then for $\psi \lor \psi'$ we take $N'' \deq N+N'$ and
the sequence ${(L''_i,R''_i)}_{i<N''}$ given by
\[
\begin{array}{l c l !{\qquad} l c l !{\qquad} l}
  L''_j
& \deq
& L_j
& R''_j
& \deq
& R_j
& \text{(for $j<N$)}
\\
  L''_{N+k}
& \deq
& L'_k
& R''_{N+k}
& \deq
& R'_k
& \text{(for $k<N'$)}
\\
\end{array}
\]

\item[Case of $(\ex x \psi)$]
Note that we can assume $x \notin V$.
We apply the induction hypothesis on $\psi$ twice:
once with $V \cup\{x\} $ and once with $V$.
This gives sequences
resp.\ ${(L_j,R_j)}_{j<N}$
and
${(L'_k,R'_k)}_{k<N'}$.
For $\ex x \psi$ we take
$N'' \deq N+N'$ and
the sequence ${(L''_i,R''_i)}_{i<N''}$ given by
\[
\begin{array}{l c l !{\qquad} l c l !{\qquad} l}
  L''_j
& \deq
& \ex x L_j
& R''_j
& \deq
& R_j
& \text{(for $j<N$)}
\\
  L''_{N+k}
& \deq
& L'_k
& R''_{N+k}
& \deq
& \ex x R'_k
& \text{(for $k<N'$)}
\\
\end{array}
\]

\noindent
The disjunction is then seen to be equivalent to $\ex x \psi$
by making a case analysis over whether $x \Leq n$ holds.

\item[Case of $(\ex X \psi)$]
Let ${(L_j,R_j)}_{j<N}$ be obtained by induction hypothesis on $\psi$.
Then for $\ex X \psi$ we keep the same $N$
and it directly follows from~(\ref{eq:app:synch:sep}) that we can take
the sequence ${(L'_j,R'_j)}_{j<N}$ given by
$L'_j \deq \exists X L_j$ and
$R'_j \deq \exists X R_j$.

\item[Case of $(\lnot \psi)$]
By induction hypothesis, we have a
natural number $N$ and two sequences of formulae
${(L_j, R_j)}_{j < N}$
such that
$\psi \longliff
\bigdisj_{j < N}
L_j\restr{[- \Leq n]} \land R_j \restr{[- \Gt n]}$.
Hence all we need to do is to add the negation,
push it through the disjunction and conjunctions using De Morgan laws and
make the disjuncts commute over the conjunction in the obtained formula.
More explicitly (leaving the parameters implicit), we have:
\begin{align*}
  \lnot \psi
&\hspace{2ex}\longliff\hspace{2ex} \bigconj\limits_{j < N}
  \lnot L_j \restr{[- \Leq n]} \lor \lnot R_j \restr{[- \Gt n]}
\\
&\hspace{2ex}\longliff\hspace{2ex} \bigdisj\limits_{f \in \two^{\{0 , \dots, N - 1 \}}}
  {\bigconj\limits_{j \in f^{-1}(0)} \lnot L_j} \restr{[- \Leq n]}
   \land {\bigconj\limits_{j \in f^{-1}(1)} \lnot R_j\restr{[- \Gt n]}}
    \tag*{\qedhere}
\end{align*}
\end{description}
\end{proof}

\begin{rem}
Note that there is a combinatorial explosion in the case of $\lnot \psi$
in Lem.~\ref{lem:app:synch:splitting} since $N_{\lnot \psi} = 2^{N_\psi}$.
It follows that the sizes of the formulae computed
in Lem.\ \ref{lem:app:synch:splitting} and the subsequent
Thm.\ \ref{thm:synch:synchbounded} are non-elementary in the size of $\psi$.
\end{rem}

\section{Automata for Atomic Formulae (\S\ref{sec:real:real})}%
\label{sec:app:real:atom}

We give below the automaton $\Sing$ (of~\S\ref{sec:real:real})
and automata for the atomic
formulae of Fig.\ \ref{fig:mso:form}.
These automata are presented as deterministic Büchi automata
(with accepting states circled).
As uniform automata, each of them has set of moves~$\one$.
Note that automata for atomic formulae involving individual variables
do not detect if the corresponding inputs actually represent natural numbers.
This is harmless, since all statements of~\S\ref{sec:real} actually assume
streams representing natural numbers to be singletons,
and since in Fig.\ \ref{fig:real:form},
quantifications over individuals are relativized to $\Sing$.
\begin{itemize}
\item $\bm{\Sing}$ \textbf{:}
\[
\begin{tikzpicture}[->,node distance=3.5cm,every state/.style={inner sep=2}]
  \node[state,initial,initial text=] (A) {$0$} ;
  \node[state,accepting] (B) [right of=A] {$1$};
  \node[state] (C) [right of=B] {$2$};

\path (A) edge [loop above] node[above, style={font=\footnotesize}] {$0$} (A) ;
\path (A) edge node[above, style={font=\footnotesize}] {$1$} (B) ;
\path (B) edge [loop above] node[above, style={font=\footnotesize}] {$0$} (B) ;
\path (B) edge node[above,style={font=\footnotesize}] {$1$} (C) ;
\path (C) edge [loop right] node[right, style={font=\footnotesize}] {$*$} (C) ;
\end{tikzpicture}
\]

\item $\bm{(x_1 \Eq x_2)}$ \textbf{:}
\[
\begin{tikzpicture}[->,node distance=3.5cm,every state/.style={inner sep=2}]
  \node[state,accepting,initial,initial text=] (A) {$0$} ;
  \node[state] (B) [right of=A] {$1$};

\path (A) edge [loop above] node[above, style={font=\footnotesize}]
  {$(i,i)$} (A) ;

\path
  (A) edge node[above, style={font=\footnotesize}] {$(i,1-i)$} (B) ;

\path (B) edge [loop right] node[right, style={font=\footnotesize}]
  {$(*,*)$} (B) ;
\end{tikzpicture}
\]

\item \bm{$(x_1 \In X_1)$} \textbf{:}
\[
\begin{tikzpicture}[->,node distance=3.5cm,every state/.style={inner sep=2}]
  \node[state,initial,initial text=] (A) {$0$} ;
  \node[state,accepting] (B) [right of=A] {$1$};

\path (A) edge [loop above] node[above, style={font=\footnotesize}]
  {$(0,*) (1,0)$} (A) ;

\path
  (A) edge node[above, style={font=\footnotesize}] {$(1,1)$} (B) ;

\path (B) edge [loop right] node[right, style={font=\footnotesize}]
  {$(*,*)$} (B) ;
\end{tikzpicture}
\]

\item \bm{$(x_1 \Leq x_2)$} \textbf{:}
\[
\begin{tikzpicture}[->,node distance=3.5cm,every state/.style={inner sep=2}]
  \node[state,initial,initial text=] (A) {$0$} ;
  \node[state] (B) [right of=A] {$1$};
  \node[state,accepting] (C) [right of=B] {$2$};

\path (A) edge [loop above] node[above, style={font=\footnotesize}]
  {$(0,*)$} (A) ;

\path (A) edge node[above, style={font=\footnotesize}] {$(1,0)$} (B) ;

\path (A) edge[bend right] node[above, style={font=\footnotesize}] {$(1,1)$} (C) ;

\path (B) edge [loop above] node[above, style={font=\footnotesize}]
  {$(*,0)$} (B) ;

\path (B) edge node[above,style={font=\footnotesize}] {$(*,1)$} (C) ;

\path (C) edge [loop right] node[right, style={font=\footnotesize}]
  {$(*,*)$} (C) ;
\end{tikzpicture}
\]

\item \bm{$\Succ(x_1,x_2)$} \textbf{:}
\[
\begin{tikzpicture}[->,node distance=2cm,every state/.style={inner sep=2}]
  \node[state,initial,initial text=] (A) {$0$} ;
  \node (a) [right of=A] {};
  \node[state] (B) [above right of=a] {$1$};
  \node[state] (C) [below right of=a] {$2$};
  \node (d) [below right of=B] {};
  \node[state,accepting] (D) [right of=d] {$3$};

\path (A) edge [loop above] node[above, style={font=\footnotesize}]
  {$(0,*)$} (A) ;

\path (A) edge node[above, style={font=\footnotesize}] {$(1,0)$} (B) ;

\path (A) edge[right] node[above, style={font=\footnotesize}] {$(1,1)$} (C) ;

\path (B) edge node[above,style={font=\footnotesize}] {$(*,1)$} (D) ;

\path (B) edge node[right, style={font=\footnotesize}] {$(*,0)$} (C) ;

\path (C) edge [loop right] node[right, style={font=\footnotesize}]
  {$(*,*)$} (C) ;

\path (D) edge [loop right] node[right, style={font=\footnotesize}]
  {$(*,*)$} (D) ;
\end{tikzpicture}
\]

\item \bm{$\Zero(x_1)$} \textbf{:}
\[
\begin{tikzpicture}[->,node distance=2cm,every state/.style={inner sep=2}]
  \node[state,initial,initial text=] (A) {$0$} ;
  \node (a) [right of=A] {};
  \node[state,accepting] (B) [above right of=a] {$1$};
  \node[state] (C) [below right of=a] {$2$};

\path (A) edge node[above, style={font=\footnotesize}] {$1$} (B) ;

\path (A) edge[right] node[above, style={font=\footnotesize}] {$0$} (C) ;

\path (B) edge [loop right] node[right, style={font=\footnotesize}] {$*$} (B) ;

\path (C) edge [loop right] node[right, style={font=\footnotesize}]
  {$*$} (C) ;
\end{tikzpicture}
\]

\item \textbf{$\True$ and $\False$:}
\begin{center}
\begin{tikzpicture}[->,node distance=3.5cm,every state/.style={inner sep=2},
  baseline=(current bounding box.center)]
\node[state,accepting,initial,initial text=] (A) {$0$} ;
\path (A) edge [loop right] node[right, style={font=\footnotesize}] {$*$} (A) ;
\end{tikzpicture}
\qquad\qquad and\qquad\qquad
\begin{tikzpicture}[->,node distance=3.5cm,every state/.style={inner sep=2},
  baseline=(current bounding box.center)]
\node[state,initial,initial text=] (A) {$0$} ;
\path (A) edge [loop right] node[right, style={font=\footnotesize}] {$*$} (A) ;
\end{tikzpicture}
\end{center}
\end{itemize}

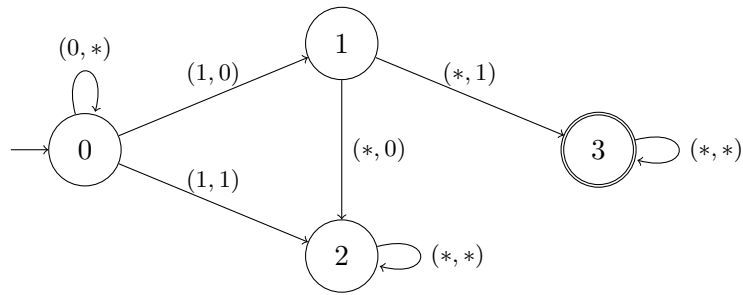
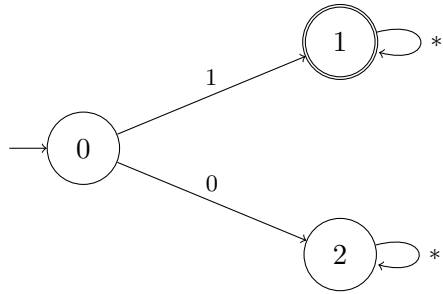
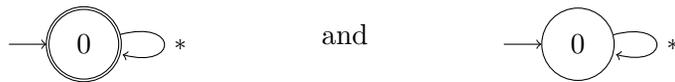
\end{document}